\documentclass[12pt]{article}

\usepackage[utf8]{inputenc}
\usepackage{amsmath,amssymb,amsthm,mathtools,amstext}

\usepackage{graphicx}
\usepackage{subcaption}
\usepackage{tikz}
\usetikzlibrary{shapes.geometric, arrows}
\usepackage{pgfplots}
\pgfplotsset{compat=1.17}
\usepackage{pgf-pie}
\usepackage{svg}
\newcommand{\Pro}{\mathbb{P}}
\usepackage[ruled,vlined]{algorithm2e}
\usepackage{algcompatible}

\usepackage{xcolor}
\usepackage{titlesec}
\newtheorem{claim}{Claim}

\newcommand{\new}[1]{\textcolor{black}{#1}}
\newenvironment{newpart}{\color{black}}{\ignorespacesafterend}

\usepackage[round,authoryear]{natbib}
\usepackage{verbatim}
\usepackage{chapterbib}
\usepackage{hyperref}
\usepackage{cancel}
\usepackage{multirow}
\usepackage{multicol}
\usepackage{eurosym}
\usepackage{academicons}
\usepackage{comment}
\usepackage{bbm}

\newtheorem{proposition}{Proposition}[section]

\newtheorem{theorem}{Theorem}
\newtheorem{algo}{Algorithm}

\newtheorem{lemma}[proposition]{Lemma}
\theoremstyle{definition}
\newtheorem{definition}[proposition]{Definition}
\newtheorem{remark}[proposition]{Remark}

\usepackage{lineno}

\begin{document}

\title{\LARGE Testing for Single-Population Ancestry in the Admixture Model}

\newcommand{\orcid}[1]{\href{https://orcid.org/#1}{\includesvg[width=10pt]{orcid}}}

\author{Holger Dette$^{(1)}$, Carola Sophia Heinzel$^{(2), (*)},$ \\
Zoe Kristin Lange$^{(1)}$, Peter Pfaffelhuber$^{(2)}$\\
(1) Department of Mathematics, Ruhr University Bochum, \\
Bochum, Germany\\
        (2)  Department of Mathematical Stochastics, University\\
        of Freiburg,  Freiburg, Germany\\
        $(*)$ Corresponding author: carola.heinzel@stochastik.uni-freiburg.de}
\date{\today}
\maketitle

\begin{abstract}
The Admixture Model describes genetic marker data by representing each individual's genome as a mixture of contributions from $K$ ancestral populations, with the individual admixture vector summarizing the corresponding ancestry proportions. In population and forensic genetics, a key question is whether an individual's genome supports a predominantly single-ancestry interpretation, while allowing for secondary ancestry contributions. We make this notion precise by calling the ancestry predominant if the largest component of the individual admixture vector exceeds a threshold above one half that is fixed by the practitioner, so that at most one population can be predominant and secondary contributions remain admissible. We propose a statistical test of this statement in the supervised Admixture Model, where ancestral allele frequencies are treated as known.

To calibrate the test, we develop a constrained parametric bootstrap procedure that generates data under a null-constrained maximum likelihood estimator, accounting for the constrained hypothesis structure, the marker-wise heterogeneity and small sample sizes. Under standard regularity conditions, we prove that the proposed test has asymptotic level $\alpha$ and is consistent, \new{so that the probability of falsely declaring predominant single-population ancestry relative to the specified reference populations is controlled asymptotically at the nominal level.}

Simulation studies demonstrate good finite-sample performance across different numbers of ancestral populations, marker-panel sizes, predominance thresholds, and allele-frequency distributions. \new{We further investigate the robustness of the procedure when the ancestral allele frequencies are estimated from finite reference panels and when the populations are only weakly differentiated.} We illustrate the practical utility of the method using data from the 1000 Genomes Project, \new{including an analysis of sensitivity to reference-panel composition.}
The proposed framework delivers interpretable, threshold-based ancestry assessment \new{with asymptotic error control relative to a specified reference panel}, and extends constrained bootstrap methodology to the independent but non-identically distributed setting of genetic marker data.
\end{abstract}

\section{Introduction}

\begin{newpart}
Genetic ancestry analysis describes how an individual's genome relates to a set of ancestral or reference populations. A widely used framework for this task is the Admixture Model, in which the allele probabilities at each marker are mixtures of the allele frequencies of $K$ ancestral populations, weighted by the individual admixture vector $q=(q_1,\dots,q_K)$. The allele frequencies $p_{kim}$ are the relative frequencies with which allele $i$ occurs at marker $m$ in population $k$, so that an allele of the individual is of type $i$ with probability $\sum_k q_k p_{kim}$. This model underlies methods such as \textsc{Structure} \citep{pritchard2000} and \textsc{Admixture} \citep{alexander2009fast}, with applications in population structure, biogeographical ancestry and forensic genetics \citep{rb2002, phillips2015, pritchard2001, tvedebrink2022, wasser2007}.

We work in the supervised setting. That is the number $K$ of reference populations and their allele frequencies are not estimated from the data at hand, but specified externally. In practice they are taken from reference databases in which the frequencies have been determined from population samples with documented geographic or ancestry information \citep{tvedebrink2022, poulsen2026}, and only the admixture vector of the individual under investigation is estimated. 

Estimating an ancestry vector is, however, only part of the problem. A practitioner must often decide whether the data provide evidence that one reference population contributes predominantly to the individual's ancestry, a decision for which hard classification methods and machine-learning approaches are commonly used \citep{wen2023forensic, phillips2009ancestry, phillips2007inferring, HEINZEL2025103290, maurer2025enhancing}. We make that notion precise. For a threshold $\varepsilon_\infty \in (1/2,1)$ chosen by the practitioner, we say that an individual has \emph{predominant single-population ancestry} if the largest component of its admixture vector,
\begin{equation*}
    d_\infty = \max_{k=1,\dots,K} q_k,
\end{equation*}
exceeds $\varepsilon_\infty$. Since $\varepsilon_\infty > 1/2$, at most one reference population can be predominant, and the threshold carries a pedigree interpretation: $\varepsilon_\infty = 0.75$ is the expected contribution of three non-admixed grandparents from the same population. Predominance is thus a statement about how much ancestry mass the leading component carries, and it explicitly permits secondary contributions from the remaining populations. Our aim is a test of this statement with a controlled probability of declaring it falsely.

Several statistical approaches investigate individual genetic ancestry, and it is worth saying precisely how the questions they answer differ from the one above. Assignment tests ask which of the candidate populations fits a genome best \citep{rannala1997}, exclusion tests ask whether a specified population can be ruled out as the origin \citep{cornuet1999}, and migrant-detection tests ask whether an individual, or one of its ancestors within a given number of generations, immigrated from elsewhere \citep{paetkau2004}. All three calibrate the decision against simulated data, a device our bootstrap calibration shares, and all three are implemented in \textsc{GENECLASS2} \citep{piry2004}. Each of them is a statement about origin rather than about the size of an ancestry component. In forensic ancestry inference, \citet{tvedebrink2018} developed likelihood-ratio tests assessing compatibility with a specified reference population, extended to two-population admixture models by \citet{tvedebrink2019inference}. Compatibility is again a different statement, since not rejecting a population bounds $d_\infty$ in neither direction.

Other procedures address the ancestry contributions themselves. Within BAPS, \citet{corander2006bayesian, corander2008enhanced} test for admixture against a null model of non-admixed ancestry, and in the supervised setting iAdmix \citep{bansal2015fast} compares models with and without particular ancestry components. Both ask whether a secondary component is present at all, whereas predominance asks how much mass the leading one carries, and the two answers are compatible in every combination: an individual may be detectably admixed and still predominantly of one ancestry, or show no detectable secondary component and yet, under a demanding threshold, fail to qualify as predominant.

Bayesian procedures provide posterior inference about recent genealogies. Under a rare-migration approximation, the prior-population-information model of \textsc{Structure} \citep{pritchard2000} assigns posterior probabilities to scenarios with at most one immigrant ancestor within a specified number of generations, and \citet{chakraborty2023efficient, chakraborty2026improved} identify parental and recent-hybrid classes while accounting for linkage and recombination. A Bayesian analysis could also address our threshold hypothesis directly, by specifying a prior on the admixture vector and computing the corresponding posterior probability, but such probabilities depend on the chosen prior and do not automatically provide the frequentist type I error control pursued here.

We test the predominance statement defined above, so that rejection provides positive evidence for a predominant contribution relative to the specified reference populations. We adapt the constrained parametric bootstrap methodology of \citet{DetteMoellenhoff} to the supervised Admixture Mode and establish pointwise asymptotic type-I error control and consistency as the number of independent, non-identically distributed markers increases. 

Section~\ref{sec:testdefinition} introduces the supervised Admixture Model and Section~\ref{section3} formulates the hypotheses, presents the bootstrap procedure and establishes its asymptotic properties. Section~\ref{sec:simulations} investigates the finite-sample performance, including estimated reference allele frequencies and weak population differentiation, Section~\ref{sec:application} applies the procedure to data from the 1000 Genomes Project and Section~\ref{sectionDiscussion} concludes.

\end{newpart}

\section{The Admixture Model}\label{sec:testdefinition}
We assume that there are $M$ markers in linkage equilibrium allowing us to assume independent occurrences of alleles at the marker loci. At each marker $m \in \{1,\dots,M\}$, $I_m$ different alleles can possibly occur. We consider $K \geq 2$ ancestral populations, whose genome is assumed to be non-admixed. For each population $k = 1, \dots, K$, allele $i$ has a frequency of $p_{kim}$ at marker $m$.
For a fixed marker $m$ and a fixed population $k$ the allele frequencies sum up to $1$, that is $\sum_i p_{kim} = 1$. Thus, the collection of ancestral allele frequencies at marker $m$ is given by $(p_{\cdot \cdot m}) \in (\mathcal{S}^{I_m})^K$, where $\mathcal{S}^{I_m}$ is the $(I_m -1)$-dimensional simplex. We work in a supervised setting, where we assume the ancestral population frequencies to be fixed and known.

The gene sample of the donor is given by $X = (X_1, \dots, X_M)$ with $X_m = (X_{im})_{i = 1,\dots,I_m}$, where $X_{im}$ is the number of copies of allele $i$ at marker $m$. The donor is assumed to belong to a diploid species, resulting in the specification that $X_{im} \in \{0,1,2\}$ with $\sum_{i} X_{im} = 2$.

We assume that the gene sample can be described by a mixture of the non-admixed genomes of the $K$ ancestral populations. In this paper, we work with a reduced $(K-1)$-dimensional parametrization of the individual admixture.
That is the individual admixture is given by $q = (q_1, \dots, q_{K-1})$, where $q_k$, $k=1,\dots,K-1$ is the probability that an allele is inherited from ancestral population $k$. Furthermore, the probability that an allele originates from ancestral population $K$ is $1 - \sum_{k=1}^{K-1} q_k$. The vector collecting the probabilities for all $K$ ancestral populations is defined as
\begin{equation*}
    q^\prime = \left(q_1, \dots, q_{K-1}, 1 - \sum_{k=1}^{K-1} q_k \right) \in \mathbb{S}^K.
\end{equation*}
For notational convenience, we introduce the notation $a^\prime = (a_1,\dots,a_{K-1}, 1- \sum_{k=1}^{K-1} a_k)$ for every vector $a = (a_1, \dots, a_{K-1}) \in \mathbb{R}^{K-1}$.
Hence, the probability of observing allele $i$ at marker $m$ given the individual admixture $q$ is $\langle q^\prime, p_{kim} \rangle$.

To establish asymptotic theory later on, we restrict the parameter space of the individual admixture to an interior subset bounded away from the boundary of the simplex.
\begin{enumerate}
    \item[$(A1)$] The parameter space is given by
    \begin{equation*}
        \Theta = \Bigg\{(q_1,\dots,q_{K-1}) \in [\varepsilon_q,1-\varepsilon_q]^{K-1} \Bigg| \varepsilon_q \leq \sum_{k=1}^{K-1} q_k \leq 1-\varepsilon_q \Bigg\}
    \end{equation*}
    for a fixed $0 < \varepsilon_q < \frac{1}{K}$.
\end{enumerate}

All specifications above lead to the following definition of the admixture model.

\begin{definition}[Admixture Model]
Let $K \geq 2$, $M \in \mathbb{N}$ and $I_m$ for all $m = 1, \dots, M$ be fixed. Let $q \in \Theta$ and $p = (p_{\cdot \cdot m})_{m = 1, \dots, M} \in \prod_{m=1}^M (\mathbb{S}^{I_m})^{K}$ be given. It holds that
\begin{equation*}
    X_{m} = (X_{1 m}, ..., X_{I_m m}) \sim \text{Multi}(2, \langle q^\prime, p_{\cdot i m} \rangle_{i=1,...,{I_m}}).
\end{equation*}
That is the vector of numbers of alleles at marker $m$ is multinomially distributed with $2$ trials and probabilities $\langle q^\prime, p_{\cdot 1 m} \rangle, \dots, \langle q^\prime, p_{\cdot I_m m} \rangle$. The family $(X_{m})_{m=1,...,M}$ is independent.
Hence, the log-likelihood function with realization $X = (X_{im})_{i = 1, \dots, I_m, m=1,\dots,M}$ scaled by $M$ is given by
\begin{align}
\begin{split}\label{Eq:LogLikelihoodFunction}
    \ell(q|X) &= \frac{1}{M} \sum_{m=1}^M \phi_m(X_{m},q)
    \\ &= \frac{1}{M} \sum_{m=1}^M \left( \log{\left(\frac{2}{X_{1 m}! \dots X_{I_m m}!}\right)} + \sum_{i=1}^{I_m} X_{i m} \log( \langle q^\prime, p_{\cdot i m} \rangle) \right),
    \end{split}
  \end{align}
\end{definition}
where $\phi_m$ denotes the single marker log-likelihood function for the $m$-th marker.
The MLE of the individual admixture $q$ is defined by
\begin{align} \label{Eq:DefinitionMLE}
         \hat q &= \text{argmax}_{q \in \Theta}\left\{q\mapsto \ell(q|X)\right\}.
\end{align}

\section{Testing Single-Population Ancestry}
\label{section3}
Firstly, we discuss the hypotheses to test for single-population ancestry in Section \ref{Section:HypothesisTest} and describe the respective algorithm for this test in Section \ref{sec:algo}. Furthermore, Section \ref{sec:algo} addresses an adapted testing procedure to test for admixed ancestry instead of single-population ancestry and gives guidance on finding interpretable thresholds for the hypotheses in a data-driven way\new{, as well as a remark on the treatment of missing data}. In Section \ref{sec:theory}, we prove that the single-population ancestry test asymptotically holds the nominal level and is consistent.

\subsection{The Hypothesis Test} \label{Section:HypothesisTest}
We investigate the individual admixture $q = (q_1, \dots, q_{K-1}) \in \Theta$ with associated full $K$-dimensional admixture vector $q^\prime \in \mathbb{S}^K$ of an individual with the gene sample $X = (X_1,\dots,X_M)$, where $X_{m} = (X_{1m}, \dots, X_{Im})$. We define the predominant component of the individual admixture as
\begin{equation*}
    d_\infty := \max_{k = 1,\ldots,K} q^\prime_k = \max \left\{ \max_{k = 1,\ldots,K-1} \{ q_k \}, 1- \sum_{k=1}^{K-1} q_k \right\}
\end{equation*}
and propose to test the hypotheses
\begin{align}\label{test:swapped}
H_0: d_\infty \leq \varepsilon_\infty
\quad \text{vs.} \quad
H_1:  d_\infty > \varepsilon_\infty
\end{align}
for a prespecified predominance threshold $\varepsilon_\infty \in (\frac{1}{2},1-(K-1)\varepsilon_q)$.

If the null hypothesis is fulfilled, the predominant component of the individual admixture is small and we conclude that the gene sample cannot be explained sufficiently well by the origin of the individual from a single ancestry population. If, on the other hand, the alternative hypothesis is fulfilled, the predominant component is large and we conclude that the individual exhibits a predominant single-population ancestry
at this threshold. The ancestry population corresponding to the index $k$ at which the maximum is obtained can be interpreted as the predominant ancestry population.

The predominance threshold $\varepsilon_\infty$  provides an interpretable benchmark of single population ancestry in terms of the expected allele contribution under an idealized pedigree model with non-admixed ancestors. For instance, $\varepsilon_\infty = 0.75$ corresponds to the expected contribution of three grandparents from the same ancestral population, and $\varepsilon_\infty = 0.875$ corresponds to the expected contribution of seven great-grandparents from the same ancestral population.
\new{These pedigree-based examples merely serve as illustrations, since our approach does not impose any restrictions on the genealogical history of an individual.}

Note that a predominance threshold smaller or equal to $\frac{1}{2}$ does not make practical sense. For $\varepsilon_\infty \leq \frac{1}{2}$ it is no evidence for single population ancestry if the individual admixture fulfills the alternative hypothesis. \new{This is an absolute definition of predominance and the question whether one component is predominant relative to the other components would require a different formulation of hypotheses.}

Furthermore, note that the restriction to the parameter space $\Theta$ in assumption $(A1)$ is a standard approach in maximum likelihood asymptotics since it ensures a regular interior regime in which the linearization of the MLE defined in \eqref{Eq:DefinitionMLE} holds uniformly and non-standard limiting distributions of the MLE are ruled out. More information on such non-standard asymptotics can be found in \cite{heinzel2025consistency}.

The hypotheses in \eqref{test:swapped} are formulated in terms of a predominance threshold $\varepsilon_\infty$ and do not consider exact purity, meaning $q^\prime_{0,k} = 1$ for one $k = 1,\dots,K$. Thus, by choosing $\varepsilon_q$ sufficiently small, the parameter space still contains all admixture vectors relevant for testing whether one population contributes at least $\varepsilon_\infty$ of the genome. Cases of exact pure ancestry lie on the boundary and are therefore outside the scope of our present interior-parameter asymptotic theory, but they can be approximated arbitrarily well within $\Theta$.

\begin{newpart}
\begin{remark}[The role and the choice of $K$ and the reference panel] \label{Rem:ChoiceOfK}

In the unsupervised Admixture Model, both the ancestral allele frequencies and the individual admixture proportions are estimated from genotype data for a specified number $K$ of ancestral populations, and selecting an appropriate value of $K$ is a well-known and challenging problem for which various methods have been developed \citep[see, for example,][]{evanno2005,wang2019parsimony,pritchard2000,raj2014,verity2016,alexander2011}. In the supervised setting considered here, by contrast, the reference allele frequencies are specified externally and $K$ is determined by the number of reference populations included.

In practice, reference allele frequencies are typically estimated from population samples with documented geographic or ancestry information \citep{tvedebrink2022}. Consequently, the hypotheses in \eqref{test:swapped} are defined relative to a specified reference panel. The null hypothesis states that no ancestry component exceeds the predominance threshold $\varepsilon_\infty$. Changing the reference populations may change the individual admixture vector $q_0$ and therefore the hypotheses themselves. Accordingly, rejection provides evidence for a predominant contribution from one of the specified reference populations, rather than identifying an individual's population of origin. The theoretical error-control guarantees apply to every fixed $K\geq 2$, provided the specified model and regularity assumptions are satisfied.

A related point concerns the adequacy of the panel. Assumption $(A4)$ is an identifiability condition on the chosen reference set and does not express whether that set is complete. If the true source population of an individual is not represented, the estimated ancestry proportions describe the genome in terms of the available groups only, and the test does not detect this by itself. Procedures that assess whether a reference database contains an appropriate population, such as the likelihood-ratio tests of \citet{tvedebrink2018, tvedebrink2019inference}, address precisely this question and are therefore complementary to, and naturally applied before, the present procedure: they examine the adequacy of the panel, whereas our test quantifies predominance within a given panel. Section \ref{sec:K4} illustrates empirically how the composition of the panel affects the results.
\end{remark}
\end{newpart}

\subsection{The Algorithm} \label{sec:algo}

We describe the algorithm for the constrained parametric bootstrap test  to test the hypotheses in \eqref{test:swapped}. The implementation of the test can be found on our \href{https://github.com/CarolaHeinzel/StatisticalTest_ConstrainedBasedBootstrap}{GitHub website}.

\begin{algo}[Constrained Parametric Bootstrap Test for Single Population Ancestry]\label{alg} \:
\newline
  \begin{itemize}
        \item[1.] Based on the sample
        \begin{equation*}
        	X = (X_1, \dots, X_M),
        \end{equation*}
        we calculate the MLE $\hat{q}$ defined in \eqref{Eq:DefinitionMLE} of the true admixture $q_0$. Furthermore, calculate the estimator of the test statistic given by
        \begin{equation*}
            \hat{d}_\infty = \max_{k=1, ..., K} \hat{q}^\prime_k.
        \end{equation*}

        \item[2.]  Define the following null-constrained estimator for the admixture.
        \begin{equation} \label{Eq:ConstrainedMLE}
                        \hat{\hat{q}} =
\begin{cases}
                \hat{q} & \text{if } \hat{d}_\infty \leq \varepsilon_\infty,\\
                \Tilde{q} & \text{if } \hat{d}_\infty >\varepsilon_\infty,
            \end{cases}
        \end{equation}
        where $\Tilde{q}$ denotes the constrained MLE
        \begin{equation*}
            \tilde q := \textup{argmax}_{q \in \Theta}\left\{q \mapsto \ell(q|X) :\max_{k=1, ..., K} q^\prime_k = \varepsilon_\infty\right\}.
        \end{equation*}
        \item[3.1] Generate bootstrap data
        \begin{align*}
        	X^{*}_{1}, \dots, X^{*}_{M}, \quad X^{*}_m = \left(X^{*}_{1, m}, ..., X^{*}_{I_m,m} \right)\sim \textup{Multi}(2, (\langle \hat{\hat q}^\prime, p_{\cdot i m}\rangle)_{i=1,\dots,I_m} )
        \end{align*}
        under the null hypothesis for an individual with admixture $\hat{\hat{q}}$.
        \item[3.2] Analogous to \eqref{Eq:DefinitionMLE}, calculate the MLE $\hat{q}^{*}$ based on the bootstrap data generated in Step 3.1. Then, calculate the bootstrap statistic
        \begin{align*}
            \hat{d}^*_\infty = \max_{k=1,...,K} \hat{q}^{*\prime}_k.
        \end{align*}
        \item[3.3] Repeat the steps $3.1$ and $3.2$ $B$ times to generate $B$ replicates of the test statistic, and denote by
        \begin{equation*}
            \hat{d}_{\infty,(1)}^{*} \leq \dots \leq \hat{d}_{\infty,(B)}^{*}
        \end{equation*}
        the order statistic of the bootstrap sample $\hat{d}^{*}_{\infty,1}, \dots, \hat{d}^{*}_{\infty,B}$.
        \item[4.] Calculate the empirical $(1-\alpha)$-quantile
        \begin{equation} \label{Eq:EmpiricalQuantileBootstrapDistribution}
            \hat{\xi}_{1-\alpha}^{B} := \hat{d}_{\infty,(\lfloor (1-\alpha) B \rfloor)}^{*}
        \end{equation}
        of the bootstrap sample.
        \item[5.] Reject the null hypothesis in \eqref{test:swapped} if
        \begin{equation*}
            \hat{d}_\infty > \hat{\xi}_{1-\alpha}^{B}.
        \end{equation*}
    \end{itemize}

\end{algo}

\begin{remark}[Testing for Admixed Ancestry]
    The suggested testing procedure in Algorithm \ref{alg} for the hypotheses in \eqref{test:swapped} asymptotically controls the type $I$ error of falsely concluding single population ancestry (see Theorem \ref{Th:BootstrapAsymptoticNormality}). If the interest of the practitioner is to control the error of falsely concluding an admixed ancestry with at least $2$ ancestry populations, swapped hypotheses can be tested. That is

    \begin{align} \label{HypothesisAdmixtureOneComponent}
        H_0^s: d_\infty  \geq \varepsilon_\infty^s
        \quad \text{vs.} \quad
        H_1^s: d_\infty < \varepsilon_\infty^s,
    \end{align}
    for a prespecified threshold $\varepsilon_\infty^s > 0$.
    A pdf-file with a detailed description of the testing procedure for the hypotheses in \eqref{HypothesisAdmixtureOneComponent}, as well as the implementation of the code can be found on our \href{https://github.com/CarolaHeinzel/StatisticalTest_ConstrainedBasedBootstrap}{GitHub website}.
\end{remark}

\begin{remark}[Sequential Testing]
    In the hypothesis testing procedure presented in this paper the choice of $\varepsilon_\infty$ in \eqref{test:swapped} is an important step. This choice has to be justified carefully in each application separately. However, there might be applications where specifying the threshold a priori is difficult. In such cases, we propose a sequential testing approach that determines an interpretable threshold from the data.
    \begin{itemize}
        \item[$1.$] Choose a grid of thresholds $\mathcal{E} = \{\varepsilon_{\infty,1}, \dots, \varepsilon_{\infty,N}\}$, starting with $\varepsilon_{\infty,1}$ close to $1$ and decreasing the thresholds, that is $\varepsilon_{\infty,n} \geq \varepsilon_{\infty,n+1}$.
        \item[$2.$] The empirical $(1-\alpha)$-quantile $\hat{\xi}_{1-\alpha}^B$ defined in \eqref{Eq:EmpiricalQuantileBootstrapDistribution} depends on the choice of $\varepsilon_{\infty}$. Calculate this quantile for each threshold in $\mathcal{E}$ by following the steps $1.$ to $4.$ in Algorithm $\ref{alg}$ and define them by $\hat{\xi}_{1-\alpha}^B(\varepsilon_{\infty,n})$ for $n = 1, \dots, N$.

        \item[3.] Define
        \begin{equation} \label{Eq:CriticalValuesSequentialTesting}
            \bar{\xi}_{1-\alpha}(\varepsilon_{\infty,n}) := \max_{k \geq n} \hat{\xi}_{1-\alpha}^B(\varepsilon_{\infty,k})
        \end{equation}
        and note that $\bar{\xi}_{1-\alpha}(\varepsilon_{\infty,n}) \geq \bar{\xi}_{1-\alpha}(\varepsilon_{\infty,n+1})$.
        \item[4.] Find the smallest $n$ for which $\hat{d}_\infty > \bar{\xi}_{1-\alpha}(\varepsilon_{\infty,n})$.
    \end{itemize}

    The smallest such index $n$ is the transition point of the sequential procedure and rejection for $\varepsilon_{\infty,n}$ implies rejection for all smaller thresholds $\varepsilon_{\infty,k}$, $k \geq n$. Consequently, $\varepsilon_{\infty,n}$ is the largest grid threshold for which the adjusted test supports predominant single-population ancestry, and all smaller thresholds are supported as well. In this sense, the procedure returns the highest predominance threshold, among the prespecified grid values, for which the data provide sequentially adjusted evidence of single-population ancestry.

    Due to the monotonicity of the critical values defined in \eqref{Eq:CriticalValuesSequentialTesting}, it follows from the sequential rejection principle discussed in \cite{Goemann} and Theorem \ref{Th:BootstrapAsymptoticNormality} that the procedure described above is an asymptotic level $\alpha$ test.

\end{remark}

\begin{newpart}
\begin{remark}[Missing data] \label{Rem:MissingData}
Missing genotypes are common in forensic applications, where degraded material frequently yields drop-out at a subset of the markers. The proposed procedure accommodates them without modification, provided the drop-out mechanism is independent of the genotype observations under the specified admixture model. In that case the observed-data log-likelihood is obtained from \eqref{Eq:LogLikelihoodFunction} by restricting the sum to the set $\mathcal{M}_{\textup{obs}} \subseteq \{1,\dots,M\}$ of successfully typed markers,
\begin{equation*}
    \ell(q|X) = \frac{1}{|\mathcal{M}_{\textup{obs}}|}
    \sum_{m \in \mathcal{M}_{\textup{obs}}} \phi_m(X_m,q),
\end{equation*}
because the markers are independent in our model and the missingness carries no information about $q_0$. Both $\hat q$ and $\tilde q$ are then computed from the observed markers only, and the bootstrap samples in step $3.1$ of Algorithm~\ref{alg} are generated for the same marker set, so the critical value adapts automatically to the reduced information and Theorem~\ref{Th:BootstrapAsymptoticNormality} continues to hold with $M$ replaced by $M_{\textup{obs}} := |\mathcal{M}_{\textup{obs}}|$, provided $M_{\textup{obs}} \to \infty$ and the design measures of the observed markers still satisfy $(A3)$ and $(A4)$.
\end{remark}
\end{newpart}

\subsection{Theoretical Results}\label{sec:theory}

In this section, we establish asymptotic correctness of Algorithm \ref{alg}. That is Theorem \ref{Th:BootstrapAsymptoticNormality} states that the constrained parametric bootstrap test described in Algorithm \ref{alg} asymptotically holds its level $\alpha$ and is consistent.

For notational convenience, we assume from now on that $I_m = I$ for all $m = 1,\dots,M$. That is we assume that the number of possible alleles is equal across all markers. We note that the theory presented here can easily be generalized for a varying number of alleles at the markers.

We make the following additional assumptions on the admixture model.
\begin{itemize}
    \item[$(A2)$] For every $k=1,\dots,K$, $i = 1,\dots, I$ and $m = 1,\dots,M$ we have that
    \begin{equation*}
        p_{kim} \geq \varepsilon_p
    \end{equation*}
    for a fixed $0 < \varepsilon_p < \frac{1}{I}$.
    \item[$(A3)$] Let $\delta_{p}$ be the Dirac measure at the point $p \in (\mathbb{S}^{I})^{K \times M}$. The fixed design measures
    \begin{equation*}
        \mu_M = \frac{1}{M} \sum_{m=1}^M \delta_{p_{\cdot \cdot m}}
    \end{equation*}
    converge weakly to an asymptotic design measure $\mu$ on $(\mathbb{S}^I)^K$.

    \item[$(A4)$] For every $h \in \mathbb{R}^K$ with $h^T \mathbf{1} = 0$ and $h \neq 0$ we have that
    \begin{equation*}
        \mu \Big( \big\{ \mathbf{p} \big\vert \exists i \text{ s.t. } h^T \mathbf{p_i} > 0 \big\} \Big) > 0.
    \end{equation*}
\end{itemize}

\begin{remark}[Discussion of assumptions]
    Assumption $(A2)$ together with assumption $(A1)$ guarantees that the log-likelihood function in \eqref{Eq:LogLikelihoodFunction} and its derivatives are uniformly bounded in both, $q$ and $p$. From an applied point of view, it can be understood as requiring that the considered marker set does not contain allele frequencies that are exactly $0$. Assumption $(A3)$ ensures that the empirical distribution of the marker-specific allele frequencies stabilizes as the number of markers increases and provides an asymptotic description of the marker distribution.

    Assumption $(A4)$ states that the asymptotic design $\mathbf{p} = (\mathbf{p}_{k i})_{k=1,\dots,K, i = 1,\dots,I}$, with $\mathbf{p_i} = (\mathbf{p_{ki}})_{k=1,\dots,K}$, is sensitive to changes in all directions $h$ on the parameter space of the ancestry population allele frequencies. Applying it to $h = q^{1 \prime} - q^{2 \prime}$, which satisfies $h^T \mathbbm{1} = 0$ and $h \neq 0$, shows that for all $q^1,q^2 \in \Theta$ with $q^1 \neq q^2$
    \begin{equation} \label{Eq:KullbackLeibler}
        \int \text{KL}\Big( \text{Mult}(2,(\langle q^2, \mathbf{p_i} \rangle )_{i=1,\dots,I}) || \text{Mult}(2,(\langle q^1, \mathbf{p_i} \rangle)_{i=1,\dots,I}) \Big) \mu(d\mathbf{p}) > 0,
    \end{equation}
    where KL denotes the Kullback-Leibler divergence. Thus, under the limiting marker design, two different individual admixtures induce different genotype distributions, so that the marker design must contain enough ancestry-informative variation to distinguish different ancestry proportions. If the ancestral populations had identical allele frequencies across all considered markers, assumption $(A4)$ would fail and the test based on these markers could not distinguish individual admixtures.
\end{remark}

We start the considerations by proving auxiliary results for the proof of Theorem \ref{Th:BootstrapAsymptoticNormality}.

We introduce the following notation. The sigma algebra generated by the random variables $X_1,\dots,X_M$ is given by
\begin{equation*}
    \mathcal{F}_M := \sigma(X_1,\dots,X_M).
\end{equation*}
The respective conditional probability,  expectation and variance given $\mathcal{F}_M$ are denoted by
\begin{equation*}
    \mathbb{P}(\cdot|\mathcal{F}_M), \: \mathbb{E}[\cdot|\mathcal{F}_M], \: \text{Var}(\cdot|\mathcal{F}_M).
\end{equation*}

The first auxiliary result states that the null-constrained MLE is consistent under the null hypothesis.
\begin{lemma} \label{Th:ConsistencyConstrainedMLE}
Let $\hat{\hat{q}}$ be the null-constrained MLE defined in \eqref{Eq:ConstrainedMLE} and assume that $d_\infty \leq \varepsilon_\infty$. If assumptions $(A1) -  (A4)$ hold it follows that
\begin{align*}
    \hat{\hat{q}} \overset{\mathbb{P}}{\longrightarrow} q_0.
\end{align*}
\end{lemma}

\begin{proof}

Firstly, we start by establishing consistency of the unconstrained MLE $\hat{q}$. This follows by Theorem $1$ in \cite{hoadley1971} and the required compactness, uniform continuity, uniform boundedness and identifiability are verified below. Similar consistency results for MLEs in the Admixture Model are discussed in \cite{pfaffelhuber2022central, heinzel2025consistency}.

Note that $\Theta$ is a compact set since it is bounded and closed as the pre-image of the closed set $[\varepsilon_q,1-\varepsilon_q]$ under the continuous function
$(x_1,\dots,x_{K-1}) \mapsto \sum_{k=1}^{K-1} x_k$.

Furthermore, the log-likelihood functions defined in \eqref{Eq:LogLikelihoodFunction} are Lipschitz continuous on $\Theta$ uniformly in $x,m$,
because for all $q^1,q^2 \in \Theta$ and all $m,x$
\begin{equation} \label{Eq:LipschitzLogLikelihood}
    |\ell_m(q^1|x) - \ell_m(q^2|x)| \leq \sup_{q \in \Theta} \|\nabla \ell_m(q|x)\| \|q^1-q^2\|.
\end{equation}
We have that $\sum_{i} x_{im} = 2$, $|p_{kim}-p_{Kim}| \leq 1$ and it follows from assumption $(A2)$ that $\langle q^{\prime},p_{\cdot im} \rangle \geq \varepsilon_p$. Thus, we obtain  for all $m,x$
\begin{equation*}
    \sup_{q \in \Theta} \|\nabla \ell_m(q|x)\| = \sup_{q \in \Theta} \left(\sum_{k=1}^{K-1} \left( \sum_{i=1}^I x_{im} \frac{p_{kim}-p_{Kim}}{\langle q^{\prime},p_{\cdot im} \rangle}\right)^2 \right)^\frac{1}{2} \leq \frac{2\sqrt{K-1}}{\varepsilon_p}.
\end{equation*}
Thus, the functions $q \mapsto \ell_m(q|X)$ are continuous on $\Theta$ uniformly in $m$ almost surely. Furthermore, these functions are uniformly bounded in $q,x$ and $m$ by assumptions $(A1)$ and $(A2)$ since
\begin{equation} \label{Eq:UniformBoundedness}
    |\ell_m(q|x)| \leq \log(2) + 2|\log(\varepsilon_p)|.
\end{equation}

Lastly, we note that by assumption $(A3)$ and $(A4)$ we have
\begin{align*}
    \lim_{M \to \infty} & \mathbb{E}[\ell_M(q)-\ell_M(q_0)] \\ &= - \int \text{KL}\Big( \text{Mult}(2,(\langle q_0, \mathbf{p_i} \rangle)_{i=1,\dots,I}) || \text{Mult}(2,(\langle q, \mathbf{p_i} \rangle)_{i=1,\dots,I} \Big) \mu(d\mathbf{p})\\ & < 0
\end{align*}
for all $q \neq q_0$.
Thus, it is straightforward to see that the assumptions of Theorem $1$ in \cite{hoadley1971} are fulfilled which implies
    \begin{equation} \label{Eq:BootstrapNormalityConsistencyMLE}
        \hat{q} \overset{\mathbb{P}}{\longrightarrow} q_0.
    \end{equation}

    To establish consistency of the     constrained MLE $\tilde{q}$, we define the parameter space
    \begin{equation*}
        \Theta_{\varepsilon_\infty} := \left\{q \in \Theta \Big| \max_{k=1,\dots,K} q^\prime_k = \varepsilon_\infty\right\}
    \end{equation*}
    and assume that $q_0 \in \Theta_{\varepsilon_\infty}$ Since $\Theta_{\varepsilon_\infty}$ is compact as well and since $\Theta_{\varepsilon_\infty} \subseteq \Theta$ the assumptions of Theorem $1$ in \cite{hoadley1971} are still fulfilled by the arguments above and it follows that
    \begin{equation} \label{Eq:BootstrapConsistencyConstrainedMLE}
        \tilde{q} \overset{\mathbb{P}}{\longrightarrow} q_0.
    \end{equation}

    To show consistency of $\hat{\hat{q}}$, we firstly consider the case $d_\infty < \varepsilon_\infty$.
    We know that
    \begin{equation} \label{Eq:FuncionDInfty}
	d_\infty:
	\left\{
	\begin{alignedat}{2}
		\Theta &\to \mathbb{R}\\
		q &\mapsto \max_{k=1,\dots,K}{ q^\prime_k }
	\end{alignedat}
	\right.
    \end{equation}
    is continuous and by the continuous mapping theorem we get
    \begin{equation} \label{Eq:StatisticEstimatorConsistency}
        \hat{d}_\infty \overset{\mathbb{P}}{\longrightarrow} d_\infty.
    \end{equation}
    Now, we have for all $\eta > 0$
    	\begin{equation} \label{Eq:Claim1Case1}
    		\Pro\left(\|\hat{\hat{q}}-q_0\|> \eta\right) \leq \Pro\left(\hat{\hat{q}} \neq \hat{q}\right) + \Pro\left(\|\hat{q}-q_0\| > \frac{\eta}{2}\right).
    	\end{equation}
        The first term on the right side of \eqref{Eq:Claim1Case1} converges to $0$ for $M \to \infty$ because it holds by \eqref{Eq:StatisticEstimatorConsistency} that
        \begin{equation*}
    		\lim_{M \to \infty} \Pro(\hat{\hat{q}} = \hat{q}) = \lim_{M \to \infty} \Pro(\hat{d}_\infty < \varepsilon_\infty) = 1.
    	\end{equation*}
        The second term on the right side of \eqref{Eq:Claim1Case1} converges to $0$ for $M \to \infty$ by consistency of the MLE established in \eqref{Eq:BootstrapNormalityConsistencyMLE}.

        Now, we consider the case $d_\infty = \varepsilon_\infty$. Note that for all $\eta > 0$
    	\begin{equation} \label{Eq:Claim1Case2}
    		\Pro\left(\|\hat{\hat{q}}-q_0\|> \eta\right) \leq \Pro\left(\|\hat{q}-q_0\| > \frac{\eta}{2}\right) + \Pro\left(\|\Tilde{q}-q_0\| > \frac{\eta}{2}\right)
    	\end{equation}
    	The first term on the right side of \eqref{Eq:Claim1Case2} converges to $0$ for $M \to \infty$ by \eqref{Eq:BootstrapNormalityConsistencyMLE} and the second term on the right side of \eqref{Eq:Claim1Case2} converges to $0$ for $M \to \infty$ by the consistency of the constrained MLE $\tilde{q}$ established in \eqref{Eq:BootstrapConsistencyConstrainedMLE}.

\end{proof}

The second auxiliary result states that the bootstrap MLE is asymptotically consistent conditionally on the data in probability for an arbitrary data generating parameter.

\begin{lemma} \label{Th:BootstrapConsistency}
Assume that $(A1)-(A4)$ hold. Let $\bar{q}$ be any $\mathcal{F}_M$-measurable sequence with values in $\Theta$ and generate bootstrap data from the admixture model with parameter $\bar{q}$. Let $\hat{q}^*$ be the corresponding bootstrap MLE. Then, it holds for all $\eta > 0$ that
    \begin{equation*}
        \mathbb{P}(\|\hat{q}^* - \bar{q} \| > \eta | \mathcal{F}_M) \overset{\mathbb{P}}{\longrightarrow} 0.
    \end{equation*}
    That is the bootstrap MLE is consistent conditionally on $\mathcal{F}_M$ in probability.
\end{lemma}

\begin{proof}
    The idea of the proof is oriented towards the proof techniques used in \cite{hoadley1971}, but extends it to the bootstrap setting.
    Let $\eta > 0$ and define
    \begin{equation*}
        \Theta(\bar{q},\eta) = \Big\{ q \in \Theta \: | \: \|q-\bar{q}\| \geq \eta \Big\}.
    \end{equation*}
    It then suffices to show that
    \begin{equation} \label{Eq:BootstrapConsistenyFirstReformulation}
        \mathbb{P}(\hat{q}^* \in \Theta(\bar{q},\eta) | \mathcal{F}_M) \overset{\mathbb{P}}{\longrightarrow} 0.
    \end{equation}
    Additionally, we define the bootstrap log-likelihood ratio for $q^1,q^2 \in \Theta$ by
    \begin{align*}
        L_M^*(q^1,q^2) &= \frac{1}{M} \sum_{m=1}^M  \left( \ell_m(q^1|X_m^*) - \ell_m(q^2|X_m^*)  \right) \\
        &= \frac{1}{M} \sum_{m=1}^M \sum_{i=1}^I X_{im}^* \Big(\log(\langle q^{1 \prime},p_{\cdot im} \rangle) - \log(\langle q^{2 \prime},p_{\cdot im} \rangle \Big),
    \end{align*}
    and
    \begin{align*}
        R_M^*(\bar{q}) = \sup_{\|q-\bar{q}\| \geq \eta} \Big\{ L_M^*(q,\bar{q}) \Big\}.
    \end{align*}

    Recall that, given $\mathcal{F}_M$, $\hat{q}^*$ maximizes $L_M^*(q, \bar{q})$ in $q$ and that $L_M^*(\bar{q},\bar{q}) = 0$. Thus, if, given $\mathcal{F}_M$, it holds that $\hat{q}^* \in \Theta(\bar{q},\eta)$, then $R_M^*(\bar{q}) \geq 0$.
    So, if we show that
    \begin{equation} \label{Eq:BootstrapConsistencyRM}
        \mathbb{P}(R_M^*(\bar{q}) \geq 0|\mathcal{F}_M) \overset{\mathbb{P}}{\longrightarrow} 0,
    \end{equation}
    then statement \eqref{Eq:BootstrapConsistenyFirstReformulation} follows and the bootstrap consistency result is proved.
    To show \eqref{Eq:BootstrapConsistencyRM}, we note that
    \begin{equation*}
        R_M^*(\bar{q}) \leq \sup_{\|q-\bar{q}\| \geq \eta} \Big\{ \big\vert L_M^*(q,\bar{q}) - \mathbb{E}\big[ L_M^*(q,\bar{q}) \big\vert \mathcal{F}_M \big] \big\vert \Big\} + \sup_{\|q-\bar{q}\| \geq \eta} \Big\{ \mathbb{E}\big[ L_M^*(q,\bar{q}) \big\vert \mathcal{F}_M \big] \Big\}
    \end{equation*}
    and thus, if $R_M^*(\bar{q}) \geq 0$, it holds for all $\gamma > 0$ that
    \begin{equation*}
        \sup_{\|q-\bar{q}\| \geq \eta} \Big\{ \big\vert L_M^*(q,\bar{q}) - \mathbb{E}\big[ L_M^*(q,\bar{q}) \big\vert \mathcal{F}_M \big] \big\vert \Big\} \geq \gamma
    \end{equation*}
    or
    \begin{equation*}
        \sup_{\|q-\bar{q}\| \geq \eta} \Big\{ \mathbb{E}\big[ L_M^*(q,\bar{q}) \big\vert \mathcal{F}_M \big] \Big\} \geq - \gamma.
    \end{equation*}
    This yields for all $\gamma > 0$ that
    \begin{align}
    \begin{split} \label{Eq:BootstrapConsistencyTwoTerms}
        \mathbb{P}(R_M^*(\bar{q}) \geq 0|\mathcal{F}_M) \leq \mathbb{P}\Big( & \sup_{\|q-\bar{q}\| \geq \eta} \Big\{ \big\vert L_M^*(q,\bar{q}) - \mathbb{E}\big[ L_M^*(q,\bar{q}) \big\vert \mathcal{F}_M \big] \big\vert \Big\} \geq \gamma | \mathcal{F}_M \Big)\\ & + \mathbbm{1}\Big\{ \sup_{\|q-\bar{q}\| \geq \eta} \Big\{ \mathbb{E}\big[ L_M^*(q,\bar{q}) \big\vert \mathcal{F}_M \big] \Big\} \geq - \gamma \Big\}
    \end{split}
    \end{align}
    and we show for both terms on the right side that they converge to $0$ in probability.
    To prove the convergence for the first term on the right side of \eqref{Eq:BootstrapConsistencyTwoTerms}, we show more generally, that a conditional uniform weak law of large numbers (WLLN) holds. That is for every $\gamma > 0$
    \begin{equation} \label{Eq:BootstrapConsistencyFirstTerm}
        \mathbb{P}\Big( \sup_{q \in \Theta} \Big\{ \big\vert L_M^*(q,\bar{q}) - \mathbb{E}\big[ L_M^*(q,\bar{q})\big\vert \mathcal{F}_M \big] \big\vert \Big\} \geq \gamma \Big\vert \mathcal{F}_M \Big) \overset{\mathbb{P}}{\longrightarrow} 0.
    \end{equation}
    To show \eqref{Eq:BootstrapConsistencyFirstTerm}, note that by the same arguments as in \eqref{Eq:UniformBoundedness} we have for every $q^1,q^2 \in \Theta$ and every $m = 1,\dots, M$
    \begin{equation*}
        \Big\vert \sum_{i=1}^I X_{im}^* \big(\log(\langle q^{1 \prime}, p_{\cdot im} \rangle)- \log(\langle q^{2 \prime}, p_{\cdot im} \rangle)\big) \Big\vert \leq 4 |\log{\varepsilon_p}|
    \end{equation*}
    almost surely.
    It follows that
    \begin{equation*}
        \text{Var}\Big( L_M^*(q^1,q^2) - \mathbb{E}\big[ L_M^*(q^1,q^2)\big\vert \mathcal{F}_M \big] \big\vert \mathcal{F}_M \Big) \leq \frac{16|\log{\varepsilon_p}|^2}{M}
    \end{equation*}
    almost surely and by a conditional version of Chebychev's inequality, we get for every $q^1,q^2 \in \Theta$ that
    \begin{equation} \label{Eq:BootstrapConsistencyPointwiseWLLN}
        \mathbb{P}\Big( \big\vert L_M^*(q^1,q^2) - \mathbb{E}\big[ L_M^*(q^1,q^2)\big\vert \mathcal{F}_M \big] \big\vert \geq \gamma \Big\vert \mathcal{F}_M \Big) \leq \frac{16|\log{\varepsilon_p}|^2}{M\gamma^2} \longrightarrow 0
    \end{equation}
    almost surely.

    By \eqref{Eq:LipschitzLogLikelihood} the pointwise conditional WLLN in \eqref{Eq:BootstrapConsistencyPointwiseWLLN} is generalized to the uniform version over the compact set $\Theta \times \Theta$ by a standard $\varepsilon$-net argument. \eqref{Eq:BootstrapConsistencyFirstTerm} follows because $\bar{q} \in \Theta$.

    To prove the convergence in probability for the second term on the right side of \eqref{Eq:BootstrapConsistencyTwoTerms}, we define
    \begin{align*}
        r(q^1,q^2) := 2\int \sum_{i=1}^I \langle q^{2 \prime}, \mathbf{p_i} \rangle \big( \log(\langle q^{1 \prime},\mathbf{p_i} \rangle) - \log(\langle q^{2 \prime},\mathbf{p_i} \rangle) \big) d\mu(\mathbf{p})
    \end{align*}
    and show that
    \begin{equation} \label{Eq:BootstrapConsistencyFunctionSet}
        \sup_{q^1,q^2 \in \Theta} \Big\vert \frac{2}{M} \sum_{m=1}^M \sum_{i=1}^I \langle q^{2 \prime},p_{\cdot im} \rangle \Big(\log(\langle q^{1 \prime},p_{\cdot im} \rangle) - \log(\langle q^{2 \prime},p_{\cdot im} \rangle \Big) - r(q^1,q^2) \Big\vert \longrightarrow 0.
    \end{equation}
    To do so, we define the set of functions
    \begin{equation*}
        \mathcal{F} = \Big\{ (\mathbb{S}^I)^K_{\varepsilon_p} \to \mathbb{R}, \: p \mapsto \sum_{i=1}^I \langle q^{2 \prime},p_i \rangle \big( \log(\langle q^{1 \prime},p_i \rangle) - \log(\langle q^{2 \prime},p_i \rangle) \big) \:\Big\vert \: q^1,q^2 \in \Theta \Big\},
    \end{equation*}
    where $(\mathbb{S}^I)^K_{\varepsilon_p} = \{ (p_{ki})_{k, i} | \sum_{i=1}^I p_{ki} = 1, \: p_{k i} \in [\varepsilon_p,1-\varepsilon_p]  \}$.
    By assumptions $(A1)$ and $(A2)$, $\mathcal{F}$ is bounded and the functions in $\mathcal{F}$ have a common Lipschitz constant, making $\mathcal{F}$ equicontinuous.
    Thus, by Theorem $1.12.1$ in \cite{Wellner}, the uniform convergence in \eqref{Eq:BootstrapConsistencyFunctionSet} holds.

    Now, we can write for every $\gamma > 0$
    \begin{align}
    \begin{split} \label{Eq:BootstrapConsistencyExpectedLikelihood}
        &\mathbb{P} \Big(\sup_{\|q-\bar{q}\| \geq \eta} \Big\{ \mathbb{E}\big[ L_M^*(q,\bar{q}) \big\vert \mathcal{F}_M \big] \Big\} \geq - \gamma \Big)\\
        &\leq \mathbb{P} \Big( \sup_{\|q-\bar{q}\| \geq \eta} \Big\{ \Big\vert \mathbb{E}\big[ L_M^*(q,\bar{q}) \big\vert \mathcal{F}_M \big] - r(q, \bar{q}) \Big\vert \Big\} \geq \gamma \Big) + \mathbb{P} \Big( \sup_{\|q-\bar{q}\| \geq \eta} \Big\{ r(q,\bar{q}) \Big\} \geq - 2 \gamma \Big).
    \end{split}
    \end{align}
    The first term on the right side of \eqref{Eq:BootstrapConsistencyExpectedLikelihood} converges to $0$ by \eqref{Eq:BootstrapConsistencyFunctionSet}. To show convergence for the second term on the right side of \eqref{Eq:BootstrapConsistencyExpectedLikelihood},
    note that
    \begin{equation*}
        \sup_{\|q-\bar{q}\| \geq \eta } \{ r(q,\bar{q}) \} \leq \sup_{(q^1,q^2) \in \Theta \times \Theta, \|q^1-q^2\|\geq \eta } \{ r(q^1,q^2) \}.
    \end{equation*}
    Since $\{ \Theta \times \Theta \: | \: \|q^1-q^2\| \geq \eta \}$ is compact and $r$ is continuous by the Leibniz integral rule, the extreme value theorem yields that $r$ attains its maximum on this set denoted by $C(\eta)$.
    Furthermore, it follows by assumption $(A4)$, that $r(q_1,q_2) < 0$ for all $q^1 \neq q^2$. Thus, we have that
    \begin{equation*}
        \sup_{\|q-\bar{q}\|\geq \eta} \{ r(q,\bar{q}) \} \leq C(\eta) < 0
    \end{equation*}
    and choosing $\gamma = -\frac{1}{4}C(\eta)$ yields that the second term on the right side of \eqref{Eq:BootstrapConsistencyExpectedLikelihood} is $0$.

\end{proof}

Now, we state and prove the main result of this article. Note that it is stated with respect to the true bootstrap quantile $\xi^*_{1-\alpha}$ as the estimation error of $\hat{\xi}^B_{1-\alpha}$ can be made arbitrarily small by increasing the number $B$ of bootstrap iterations.

\begin{theorem} \label{Th:BootstrapAsymptoticNormality}
	Let assumptions $(A1) - (A4)$ hold. Let $\Theta^\circ$ be the interior of the set $\Theta$ from assumption $(A1)$ and assume $q_0 \in \Theta^\circ$.

    \begin{itemize}
        \item[(a)] If the null hypothesis $H_0$ in \eqref{test:swapped} holds, it follows for any $\alpha \in (0,0.5)$ that
        \begin{equation*}
        \lim_{M \to \infty} \Pro(\hat{d}_\infty > \xi_{1-\alpha}^*)
        \begin{cases}
            = 0 & \text{if } d_\infty < \varepsilon_\infty,\\
            = \alpha & \text{if } d_\infty = \varepsilon_\infty.
        \end{cases}
        \end{equation*}
        \item[(b)] If the alternative hypothesis $H_1$ in \eqref{test:swapped} holds, it follows for any $\alpha \in (0,0.5)$ that
        \begin{equation*}
            \lim_{M \to \infty} \Pro(\hat{d}_\infty > \xi_{1-\alpha}^*) = 1.
        \end{equation*}
    \end{itemize}
\end{theorem}

\begin{proof}[Proof of Theorem \ref{Th:BootstrapAsymptoticNormality}] ~~

\smallskip

\noindent
{\bf Proof of part(a).}  The proof  is separated in $3$ steps and draws on ideas of the proofs of Theorem $2$ in \cite{DetteMoellenhoff} and Theorem $2$ in \cite{hoadley1971}.

\bigskip

\begin{claim} \label{Claim2}
It holds that
    \begin{equation*}
    		\sqrt{M} \left( \hat{q}^{*} - \hat{\hat{q}} \right)
    \end{equation*}
    converges in distribution to $\mathcal{N}(0,\Gamma^{-1}(q_0))$ conditionally on $\mathcal{F}_M$ in probability with information matrix
    \begin{align}
    \begin{split} \label{Eq:BootstrapNormalityGamma}
        \Gamma(q_0) = \lim_{M \to \infty} \frac{1}{M} \sum_{m=1}^M \mathbb{E}_{q_0}\left[ \dot{\phi}_m(X_m,q_0) \dot{\phi}_m(X_m,q_0)^T \right]
    \end{split}
    \end{align}
    where $\dot{\phi}_m$ is the gradient with respect to the parameter $q$ of the single-marker log-likelihood function $\phi_m$ given in \eqref{Eq:LogLikelihoodFunction}
\end{claim}

\begin{proof}[Proof of Claim 1.] \renewcommand{\qedsymbol}{\ensuremath{\triangle}}
    Firstly, we show that $\Gamma$ and $\Gamma^{-1}$ are well-defined for the function
    \begin{align*}
        \Gamma(q) &= \lim_{M \to \infty} \frac{1}{M} \sum_{m=1}^M \mathbb{E}_{q}\left[ \dot{\phi}_m(X_m,q) \dot{\phi}_m(X_m,q)^T \right]\\ &= \lim_{M \to \infty} \frac{1}{M} \sum_{m=1}^M \mathbb{E}_{q}\left[ -\ddot{\phi}_m(X_m,q) \right],
    \end{align*}
    where $\ddot{\phi}_m(x,q)$ is the second derivative of the function $\phi_m$ defined in \eqref{Eq:LogLikelihoodFunction}.
    By assumption $(A3)$ the limit of the information matrix estimator exists for every $q \in \Theta$ and is given by
    \begin{align*}
        \Gamma(q) &= \Big( 2 \int \sum_{i=1}^I \frac{(\mathbf{p_{k i}}-\mathbf{p_{K i}})(\mathbf{p_{\ell i}}-\mathbf{p_{K i}})}{\langle \mathbf{p_{\cdot i}}, q^\prime \rangle} \mu(d \mathbf{p}) \Big)_{k,\ell = 1,\dots,K-1}.
    \end{align*}
    By assumption $(A4)$, $\Gamma(q)$ is positive definite for every $q \in \Theta^\circ$. To see this, note that for every $h \in \mathbb{R}^{K-1}$ with $h \neq 0$ we have
    \begin{equation*}
        h^T \Gamma(q) h = 2 \int \sum_{i=1}^I \frac{1}{\langle \mathbf{p_{\cdot i}}, q^\prime \rangle} \Big( \sum_{k=1}^{K-1} h_k (\mathbf{p_{k i}}-\mathbf{p_{K i}}) \Big)^2 \mu (d \mathbf{p}).
    \end{equation*}
    This term is larger than $0$ iff $\sum h_k (\mathbf{p_{k i}}-\mathbf{p_{K i}}) \neq 0$ on a set with positive measure and this statement is equivalent to assumption $(A4)$.

    We show for the bootstrap MLE $\hat{q}^*$ that for all $\eta > 0$
    \begin{equation*}
        \mathbb{P} \left(  \| \hat{q}^* - q_0 \|  > \eta | \mathcal{F}_M  \right) \overset{\mathbb{P}}{\longrightarrow } 0.
    \end{equation*}
    To see this, note that we have for all $\eta, \delta > 0$
    \begin{align} \label{Eq:ExistenceBootstrapMLE}
        \begin{split}
        \mathbb{P} & \left( \mathbb{P} \left(  \| \hat{q}^* - q_0 \|  > \eta | \mathcal{F}_M \right) > \delta  \right)\\ &\leq \mathbb{P} \left( \mathbb{P} \left(  \| \hat{q}^* - \hat{\hat{q}} \|  > \frac{\eta}{2} \bigg\vert \mathcal{F}_M \right) > \frac{\delta}{2} \right) + \mathbb{P}(\|\hat{\hat{q}} - q_0\| > \frac{\eta}{2}).
        \end{split}
    \end{align}
    The first term on the right side of \eqref{Eq:ExistenceBootstrapMLE} converges to $0$ by Lemma \ref{Th:BootstrapConsistency} and the second term converges to $0$ by Lemma \ref{Th:ConsistencyConstrainedMLE}.

    Since $\Theta^\circ$ is an open subset of $\mathbb{R}^{K-1}$ there exists a constant $\eta > 0$ such that  $B(q_0,\eta) \subsetneq \Theta^\circ$, where
    $B(q_0,\eta)$ denote   the open $(K-1)-$ball with radius $\eta$ centered at $q_0$.
    Then,
    \begin{equation*}
        \mathbb{P}(\hat{q}^* \in B(q_0,\eta)|\mathcal
        {F}_M) \overset{\mathbb{P}}{\longrightarrow} 1.
    \end{equation*}
    So on the high conditional-probability event $\{\hat{q}^* \in B(q_0,\eta)\}$, it holds that
    \begin{equation*}
        \sum_{m=1}^M \dot{\phi}_m(X_m^*,\hat{q}^*) = 0.
    \end{equation*}
    Analogously to \cite{hoadley1971}, we have
    \begin{equation} \label{Eq:LinearizationBootstrapEstimator}
        \frac{1}{\sqrt{M}} \sum_{m=1}^M \dot{\phi}_m(X_m^*,\hat{\hat{q}}) = I_M^* [\sqrt{M}(\hat{q}^*-\hat{\hat{q}})],
    \end{equation}
    where
    \begin{equation*}
        I_M^* = \int_0^1 \frac{1}{M} \sum_{m=1}^M \Big[-\ddot{\phi}_m\Big(X_m^*,\hat{\hat{q}}+\xi(\hat{q}^*-\hat{\hat{q}})\Big)\Big] d\xi.
    \end{equation*}
    We show the following statement componentwise.
    \begin{equation} \label{Eq:Claim2IM}
        I_M^* \overset{\mathbb{P}}{\longrightarrow}\Gamma(q_0) \text{ conditionally on } \mathcal{F}_M \text{ in probability,}
    \end{equation}
    where $\Gamma(q_0)$ is defined in \eqref{Eq:BootstrapNormalityGamma}.
   For this purpose, note that every component $(\ddot{\phi}_m)_{k \ell}$ is Lipschitz in $q$ with a Lipschitz constant independent of $k,\ell,x$ and $m$. This statement follows, because for all $k,\ell,r=1,\dots,K-1$
    \begin{align*}
        \Big\vert \frac{\partial}{\partial q_r} (\ddot{\phi}_m(x,q))_{k \ell} \Big\vert &= \Big\vert 2 \sum_{i=1}^I \frac{(p_{kim}-p_{Kim})(p_{\ell im}-p_{Kim})(p_{rim}-p_{Kim}) x_{im}}{\langle q^\prime, p_{\cdot im} \rangle^3} \Big\vert
        \leq \frac{4}{\varepsilon_p^3}
    \end{align*}
    since $\vert p_{kim}-p_{Kim} \vert \leq 1$, $\sum_{i=1}^I x_{im} = 2$ and $\langle q^\prime, p_{\cdot im} \rangle \geq \varepsilon_p$. Thus, we have uniformly in $x$ and $m$ that for all $q_1,q_2 \in \Theta$
    \begin{align}
    \begin{split} \label{Eq:LipschitzContinuityHessianPhi}
        \Big \vert (\ddot{\phi}_m(x,q_1))_{k \ell} - (\ddot{\phi}_m(x,q_2))_{k \ell} \Big \vert &\leq \| \nabla_q (\ddot{\phi}_m(x,q))_{k \ell} \| \| q_1 - q_2 \| \\ &\leq
        L_{\ddot{\phi}} \| q_1 - q_2 \|,
    \end{split}
    \end{align}
    where $L_{\ddot{\phi}}:=\frac{4\sqrt{K-1}}{\varepsilon_p^3} $.
    Now, we define
    \begin{align*}
        J_M^*(\hat{\hat{q}}) &= \frac{1}{M} \sum_{m=1}^M -\ddot{\phi}_m(X^*_m, \hat{\hat{q}}), \\
        J_M(\hat{\hat{q}}) &= \frac{1}{M} \sum_{m=1}^M \mathbb
        E[ -\ddot{\phi}_m(X^*_m, \hat{\hat{q}}) | \mathcal{F}_M ],
    \end{align*}
  use  the decomposition
    \begin{align}
     \label{Eq:DecompositionIM}
        &\vert (I_M^*)_{k \ell} - (\Gamma(q_0))_{k \ell} \vert \\
        &\leq \vert (I_M^*)_{k \ell} - (J_M^*(\hat{\hat{q}}))_{k \ell} \vert + \vert (J_M^*(\hat{\hat{q}}))_{k \ell} - (J_M(\hat{\hat{q}}))_{k \ell} \vert + \vert (J_M(\hat{\hat{q}}))_{k \ell} - (\Gamma(q_0))_{k \ell} \vert
        \nonumber
    \end{align}
    and show the convergence of each of the three summands separately.
    For the first term on the right side of \eqref{Eq:DecompositionIM}, we infer with the Lipschitz continuity derived in \eqref{Eq:LipschitzContinuityHessianPhi} that
    \begin{align*}
        &\vert (I_M^*)_{k \ell} - (J_M^*(\hat{\hat{q}}))_{k \ell} \vert \\
        &\leq \int_0^1 \frac{1}{M} \sum_{m=1}^M \vert (\ddot{\phi}_m(X_m^*, \hat{\hat{q}} + \xi (\hat{q}^* - \hat{\hat{q}})) )_{k \ell} - (\ddot{\phi}_m(X_m^*, \hat{\hat{q}}) )_{k \ell} \vert \: d\xi \\
        &\leq \int_0^1 \frac{1}{M} \sum_{m=1}^M L_{\ddot{\phi}} \| \hat{\hat{q}} + \xi (\hat{q}^* - \hat{\hat{q}}) - \hat{\hat{q}} \| \: d\xi \\
       & = L_{\ddot{\phi}} \| \hat{q}^* - \hat{\hat{q}} \|
    \end{align*}
    and the last term converges to $0$ in probability
    conditionally on $\mathcal{F}_M$ in probability by Lemma \ref{Th:BootstrapConsistency}.

    For the second term on the right side of \eqref{Eq:DecompositionIM}, we have that
    \begin{align*}
        &\vert (J_M^*(\hat{\hat{q}}))_{k \ell} - (J_M(\hat{\hat{q}}))_{k \ell} \vert \\
        &= \Big\vert \frac{1}{M} \sum_{m=1}^M \Big[ (\ddot{\phi}_m(X_m^*, \hat{\hat{q}}) )_{k \ell} - \mathbb{E}[\ddot{\phi}_m(X_m^*, \hat{\hat{q}}) | \mathcal{F}_M]_{k \ell} \Big] \Big\vert
    \end{align*}
    converges in probability conditionally on $\mathcal{F}_M$ in probability to 0 by a conditional law of large numbers. This holds by the conditional Chebychev inequality since
    \begin{equation*}
        \sup_{m} \text{Var}(\ddot{\phi}(X_m^*,\hat{\hat{q}})|\mathcal{F}_M) \leq \frac{4}{\varepsilon_p^4}
    \end{equation*}
    almost surely.
    To prove the convergence for the third term in \eqref{Eq:DecompositionIM}, we show that
    \begin{align}
    \begin{split} \label{Eq:DecompositionThirdTermIM}
        \vert (J_M(\hat{\hat{q}}))_{k \ell} - (\Gamma(q_0))_{k \ell} \vert
        \leq \sup_{q \in \Theta} &\Big \vert \frac{2}{M} \sum_{m=1}^M \sum_{i=1}^I \frac{(p_{kim}-p_{Kim})(p_{\ell im}-p_{Kim})}{\langle q^\prime,p_{\cdot im} \rangle} \\ &- 2 \int \sum_{i=1}^I \frac{(\mathbf{p}_{ki}-\mathbf{p}_{Ki})(\mathbf{p}_{\ell i}-\mathbf{p}_{Ki})}{\langle q^\prime,\mathbf{p}_{\cdot i} \rangle} \mu (d \mathbf{p}) \Big \vert\\
        & + \vert \Gamma(\hat{\hat{q}})_{k \ell} - \Gamma(q_0)_{k \ell} \vert \overset{\mathbb{P}}{\longrightarrow} 0.
    \end{split}
    \end{align}
    The second term on the right side of the inequality in \eqref{Eq:DecompositionThirdTermIM} converges to $0$ in probability by the continuous mapping theorem together with the continuity of $\Gamma$ and Lemma \ref{Th:ConsistencyConstrainedMLE}. To prove the convergence for the first term on the right side of the inequality in \eqref{Eq:DecompositionThirdTermIM}, we define
    \begin{equation*}
        F = \{ g_q ~| ~q \in \Theta \}
    \end{equation*}
    with
    \begin{equation*}
	g_q:
	\left\{
	\begin{alignedat}{2}
		(\mathbb{S}^I)^K_{\varepsilon_p} &\to \mathbb{R}\\
		p &\mapsto \sum_{i=1}^I \frac{(p_{ki} - p_{Ki})(p_{\ell i} - p_{Ki})}{\langle q^\prime,p_i \rangle}.
	\end{alignedat}
	\right.
\end{equation*}
and note that $F$ is equicontinuous and uniformly bounded as a set of continuous and bounded functions on a compact space. By Corollary $11.3.4$ in \cite{Dudley_2002}, we get
\begin{equation*}
    \lim_{M \to \infty} \sup_{q \in \Theta} \left| \int g_q(p) \: d(\mu_M - \mu)(p) \right| = 0
\end{equation*}
proving the desired convergence and completing the proof of \eqref{Eq:Claim2IM}.

We continue applying a central limit theorem to the left hand side of \eqref{Eq:LinearizationBootstrapEstimator}.
For this purpose we note that the condition of the conditional Ljapunov theorem is satisfied for every linear combination of the components of $\dot{\phi}_m(X_m^*,\hat{\hat{q}})$, because
    \begin{equation*}
		\frac{1}{M^{\frac{3}{2}}} \sum_{m=1}^{M} \mathbb{E}\left[\left| \dot{\phi}_{m,k}(X_m^*,\hat{\hat{q}}) \dot{\phi}_{m,\ell}(X_m^*,\hat{\hat{q}}) \dot{\phi}_{m,r}(X_m^*,\hat{\hat{q}})\right|\bigg| \mathcal{F}_M \right] \leq \frac{8 }{\varepsilon_p^3 \sqrt{M}} \longrightarrow 0~~
	\end{equation*}
    almost surely for all $k, \ell, r = 1,\dots, K-1$.
    Thus, it follows by the conditional Ljapunov central limit theorem and the conditional Cramér-Wold theorem that
    \begin{equation*}
        J_M(\hat{\hat{q}})^{-\frac{1}{2}}\frac{1}{\sqrt{M }} \sum_{m=1}^M \dot{\phi}_m(X_m^*,\hat{\hat{q}})
    \end{equation*}
    converges to a multivariate standard normal distribution conditionally on $\mathcal{F}_M$ in probability.

    By \eqref{Eq:Claim2IM}, \eqref{Eq:DecompositionThirdTermIM} and the conditional slutsky theorem we finally obtain  that
    \begin{equation*}
        (I_M^*)^{-1} \frac{1}{\sqrt{M}} \sum_{m=1}^M \dot{\phi}_m(X_m^*,\hat{\hat{q}})
    \end{equation*}
    converges to a multivariate normal distribution with expected value $0$ and covariance matrix $\Gamma^{-1}(q_0)$ conditionally on $\mathcal{F}_M$ in probability. This proves the statement of Claim \ref{Claim2} since the equality in \eqref{Eq:LinearizationBootstrapEstimator} holds with high conditional probability.

    \end{proof}

    \begin{claim} \label{Claim3}
        It holds for any $\alpha \in (0,0.5)$ that
    \begin{equation*}
        \lim_{M \to \infty} \mathbb{P}(\hat{d}_\infty > \xi^*_{1-\alpha}) = 0
    \end{equation*}
    if $d_\infty < \varepsilon_\infty$.
    \end{claim}

\begin{proof}[Proof of Claim 2.] \renewcommand{\qedsymbol}{\ensuremath{\triangle}}

    Recall that $\xi_{1-\alpha}^*$ is the $(1-\alpha)$-quantile of the conditional distribution of $\hat{d}^*_\infty$ and define
    \begin{equation*}
        \hat{\hat{d}}_\infty = \max_{k=1}^K \hat{\hat{q}}^\prime_k
    \end{equation*}
    It follows that
	\begin{align}
    \begin{split} \label{Eq:Claim3Decomposition}
		\mathbb{P}(\hat{d}_\infty > \xi_{1-\alpha}^*) &= \mathbb{P}(\hat{d}_\infty > \xi_{1-\alpha}^* \land \hat{d}_\infty \leq \varepsilon_\infty) + \mathbb{P}(\hat{d}_\infty > \xi_{1-\alpha}^* \land \hat{d}_\infty > \varepsilon_\infty) \\
		&\leq \mathbb{P}(\hat{\hat{d}}_\infty > \xi_{1-\alpha}^*) + \mathbb{P}(\hat{d}_\infty > \varepsilon_\infty)
    \end{split}
	\end{align}
	We show for both terms on the right side of in \eqref{Eq:Claim3Decomposition} that they converge to $0$ for $M \to \infty.$
   For the second term on the right side of \eqref{Eq:Claim3Decomposition} this statement is obvious  because $\hat{d}_\infty \overset{\mathbb{P}}{\longrightarrow} d_\infty < \varepsilon_\infty$.
    To infer the convergence for the first term,  note that
    \begin{equation*}
        p_{1-\alpha}^* = \sqrt{M}(\xi_{1-\alpha}^* - \hat{\hat{d}}_\infty)
    \end{equation*}
    is the conditional $(1-\alpha)-$quantile of the statistic $\sqrt{M}(\hat{d}^*_{\infty}-\hat{\hat{d}}_\infty)$.
    Consequently,
    \begin{equation*}
        \mathbb{P}(\hat{\hat{d}}_\infty > \xi_{1-\alpha}^*) = \mathbb{P}(p_{1-\alpha}^* < 0)
    \end{equation*}
    and in order to show Claim \ref{Claim3} it is enough to prove
    \begin{equation} \label{Eq:Claim3QuantileNegative}
        \mathbb{P}(p_{1-\alpha}^* < 0) \longrightarrow 0
    \end{equation}
    as  $M \to \infty$.
    To do so, choose any $k_M \in \arg \max_k \hat{\hat{q}}^\prime_k$. Then, conditionally on $\mathcal{F}_M$
    \begin{equation} \label{Eq:Claim3TM*}
        \sqrt{M}(\hat{d}_\infty^* - \hat{\hat{d}}_\infty) \geq \sqrt{M} (\hat{q}_{k_M}^{* \prime} - \hat{\hat{q}}_{k_M}^\prime) =: T_{M, k_M}^*
    \end{equation}
    almost surely. For a fixed $k$, we have
    \begin{equation} \label{Eq:Claim3T}
         T_{M,k}^* \overset{\mathcal{D}}{\longrightarrow} T_k \sim \mathcal{N}(0, a_k^T \Gamma^{-1}(q_0)a_k)
    \end{equation}
    conditionally on $\mathcal{F}_M$ in probability,
    where
    \begin{equation} \label{Eq:ak}
        a_k = \begin{cases}
             e_{k} & \text{if } k \neq K, \\
             (-1, \dots, -1)^T & \text{if } k = K.
        \end{cases}
    \end{equation}
    Note that for all $k = 1,\dots,K$ $a_k^T \Gamma^{-1}(q_0)a_k > 0$ by assumption $(A4)$.

    Since $\alpha \in (0,0.5)$  by assumption, the $(1-\alpha)-$quantile of $T_k$, say $\xi_{1-\alpha}^{T_k},$ is positive. Consequently, we can choose $\delta > 0$ and $c > 0$ such that
    \begin{equation} \label{Eq:FkBound}
        \max_{k=1}^K F_k(c) < 1- \alpha - \delta,
    \end{equation}
    where $F_k$ is the distribution function of $T_k$.
     We denote by $F_M^*$ the distribution function of $\sqrt{M}(\hat{d}_\infty^* - \hat{\hat{d}}_\infty)$ conditionally on $\mathcal{F}_M$, by $F_{M,k}$ the distribution function of $T_{M,k}^*$ conditionally on $\mathcal{F}_M$. By \eqref{Eq:Claim3T} and the continuity of $F_k$, we get
     \begin{equation} \label{Eq:Claim3UniformConvergenceFM}
         \max_{k=1}^K \sup_{t \in \mathbb{R}} |F_{M,k}(t)-F_k(t)| \overset{\mathbb{P}}{\longrightarrow} 0
     \end{equation}
     as $M \to \infty$.
     It follows from \eqref{Eq:Claim3TM*} that
     \begin{equation*}
         F_M^*(c) \leq F_{M,k_M}(c) \leq \max_{k=1}^K F_{M,k}(c)
     \end{equation*}
     almost surely
     and together with \eqref{Eq:Claim3UniformConvergenceFM} and the uniform continuity of $F_k$, this yields for any  $\gamma > 0$
     \begin{equation} \label{Eq:Claim3FM*F}
         \mathbb{P}(F_M^*(c) \leq \max_{k=1}^K F_k(c) + \gamma) \longrightarrow 1
     \end{equation}
     as  $M \to \infty$. By setting $\gamma = \frac{\delta}{2}$ in \eqref{Eq:Claim3FM*F}, we obtain with \eqref{Eq:FkBound} that
     \begin{equation*}
         \mathbb{P} \Big (F_M^*(c) > 1-\alpha-\frac{\delta}{2} \Big ) \longrightarrow 0.
    \end{equation*}
    If $F_M^*(c) < 1-\alpha-\frac{\delta}{2}$, then $p_{1-\alpha}^* > c > 0$,
    which proves  \eqref{Eq:Claim3QuantileNegative} because
    \begin{align*}
        \mathbb{P}\Big(p_{1-\alpha}^* < 0\Big) &\leq \mathbb{P}\Big(F_M^*(c) \geq 1- \alpha- \frac{\delta}{2}\Big) \longrightarrow 0.
    \end{align*}

    \end{proof}

    \begin{claim} \label{Claim4}
        It holds for any $\alpha \in (0,0.5)$ that
    \begin{equation*}
        \lim_{M \to \infty} \mathbb{P}(\hat{d}_\infty > \xi^*_{1-\alpha}) = \alpha
    \end{equation*}
    if $d_\infty = \varepsilon_\infty$.
    \end{claim}

\begin{proof}[Proof of Claim 3.] \renewcommand{\qedsymbol}{\ensuremath{\triangle}}

    Since $\varepsilon_\infty \in (\frac{1}{2},1)$ by assumption, we know that
    \begin{equation*}
        \arg \max \{q_{0,k}: k=1,\dots,K\} = \{k_0\}
    \end{equation*}
    is a singleton. Thus, it follows by \cite{Carcamo} that
    \begin{equation*}
	d_\infty:
	\left\{
	\begin{alignedat}{2}
		\Theta &\to \mathbb{R}\\
		q &\mapsto \max_{k=1}^K q^\prime_k
	\end{alignedat}
	\right.
\end{equation*}
is Hadamard differentiable in $q_0$ with derivative $d^{\prime}_{\infty, q_0}(h) = a_{k_0}^T h$ with $a_{k_0}$ defined in \eqref{Eq:ak}. Then, it follows by Theorem $23.9$ (Delta Method for Bootstrap) in \cite{vanDerVaart} and Claim \ref{Claim2} that
\begin{equation}
    \sqrt{M} (\hat{d}^*_\infty - \hat{\hat{d}}_\infty) \overset{\mathcal{D}}{\longrightarrow} T_{k_0} \sim \mathcal{N}(0, a_{k_0}^T \Gamma^{-1}(q_0) a_{k_0})
\end{equation}
conditionally on $\mathcal{F}_M$ in probability.

Note that
\begin{equation*}
    \rho_{1-\alpha}^* := \frac{\sqrt{M}(\xi_{1-\alpha}^*-\hat{\hat{d}}_\infty)}{\sigma_{k_0}}
\end{equation*}
with $\sigma_{k_0}^2 = a_{k_0}^T \Gamma^{-1}(q_0) a_{k_0} > 0$ is the conditional $(1-\alpha)$-quantile of $\sqrt{M}(\hat{d}_\infty^*- \hat{\hat{d}}_\infty)/\sigma_{k_0}$. With Lemma $21.2$ in \cite{vanDerVaart} we get
\begin{equation} \label{Eq:Claim4QuantileConvergence}
    \rho_{1-\alpha}^* \overset{\mathbb{P}}{\longrightarrow} u_{1-\alpha},
\end{equation}
where $u_{1-\alpha}$ is the $(1-\alpha)$-quantile of the standard normal distribution.
Thus, for every $\alpha \in (0,0.5)$, it holds by \eqref{Eq:Claim4QuantileConvergence} that
\begin{align*}
    \mathbb{P}(\hat{\hat{d}}_\infty > \xi_{1-\alpha}^*) = \mathbb{P}(\rho_{1-\alpha}^* < 0) \longrightarrow 0.
\end{align*}
This statement is equivalent to statement $(A.13)$ in \cite{DetteMoellenhoff}. By replacing $<$ with $>$ and $\geq$ with $\leq$ in the proof of Theorem $2$ in \cite{DetteMoellenhoff} it follows with the same arguments that
\begin{equation} \label{Eq:Claim4FinalEquality}
    \lim_{M \to \infty} \mathbb{P}(\hat{d}_\infty > \xi_{1-\alpha}^*) = \lim_{M \to \infty} \mathbb{P} \Big(\frac{\sqrt{M} (\hat{d}_\infty - d_\infty)}{\sigma_{k_0}} > \frac{\sqrt{M}(\xi_{1-\alpha}^* - \hat{\hat{d}}_\infty)}{\sigma_{k_0}}\Big).
\end{equation}
By \begin{equation*}
    \sqrt{M}(\hat{q}-q_0) \overset{\mathcal{D}}{\longrightarrow} \mathcal{N}(0,\Gamma^{-1}(q_0)),
\end{equation*}
that holds by the same arguments given in the proof of Claim \ref{Claim2},
and \eqref{Eq:Claim4QuantileConvergence} it holds that the right side of \eqref{Eq:Claim4FinalEquality} converges to $1-\Phi(u_{1-\alpha}) = \alpha$, where $\Phi$ denotes the cumulative distribution function of the standard normal distribution.
 \end{proof}

 \smallskip

\noindent
{\bf Proof of part(b).}
 Recall that we know by the continuous mapping theorem that
    \begin{equation*}
        \hat{d}_\infty \overset{\mathbb{P}}{\longrightarrow} d_\infty > \varepsilon_\infty.
    \end{equation*}
    Thus, there exists $\delta > 0$, such that
    \begin{equation*}
        \mathbb{P}(\hat{d}_\infty > \varepsilon_\infty + 2\delta) \longrightarrow 1.
    \end{equation*}
    On this high probability event, we have $\hat{\hat{q}} = \tilde{q}$ and $\hat{\hat{d}}_\infty = \varepsilon_\infty$.
    By applying Lemma \ref{Th:BootstrapConsistency} with $\bar{q} := \hat{\hat{q}}$, we get
    \begin{equation*}
        \mathbb{P}(\|\hat{q}^* - \hat{\hat{q}} \| > \eta | \mathcal{F}_M) \overset{\mathbb{P}}{\longrightarrow} 0
    \end{equation*}
    and by the Lipschitz continuity of $d_\infty(q)$ defined in \eqref{Eq:FuncionDInfty}, it follows that
    \begin{equation*}
        \mathbb{P}(\hat{d}_\infty^* > \varepsilon_\infty + \delta | \mathcal{F}_M) \overset{\mathbb{P}}{\longrightarrow} 0.
    \end{equation*}
    Hence,
    \begin{equation*}
        \mathbb{P}(\xi_{1-\alpha}^* \leq \varepsilon_\infty + \delta) \longrightarrow 1.
    \end{equation*}
    Finally,
    \begin{align*}
        \mathbb{P}(\hat{d}_\infty > \xi^*_{1-\alpha}) \geq \mathbb{P}(\hat{d}_\infty > \varepsilon_\infty + 2 \delta, \: \xi^*_{1-\alpha} \leq \varepsilon_\infty + \delta) \longrightarrow 1.
    \end{align*}

    \end{proof}

\section{Empirical Evaluation of the Statistical Test}\label{sec:simulations}

To empirically evaluate the statistical testing procedure described in Algorithm \ref{alg}, we simulate genome data under the Admixture Model and analyze the test's performance on simulated data for which we know the ground truth.

\subsection{General Evaluation}
Firstly, we generate the ancestral population allele frequencies by setting $I = 2$ and sampling $p_{k1m}$ for $k=1,\dots,K$ and $m=1,\dots,M$ i.i.d. from either a  $\mathcal U(0, 1)$-distribution or a $\beta(0.5, 2)$-distribution. Then we set $p_{k2m} = 1-p_{k1m}$. This results in scenarios with either $p_{k \cdot m} \sim $ Dirichlet$(1,1)$ or $p_{k \cdot m} \sim $ Dirichlet$(0.5,2)$.

We define the true individual admixture by $q_0$ and the true value of the test statistic by
\begin{equation*}
    T := \max_{k=1,\dots,K} q^\prime_{0,k}
\end{equation*}
with $q^\prime = (q_1,\dots,q_{K-1},1-\sum_{k} q_k)$. We sample $q_0$ by first setting $q_k = T$ for $k \sim \mathcal U(\{1,...,K\}).$ Then, for the remaining $K-1$ components, we sample $(q'_1,..., q'_{k-1}, q'_{k+1}, q'_K) \sim$ Dirichlet$((1,...,1)) \in \mathbb R^{K-1}$. We finally normalize the true ancestral vector. That is we have
$$q_0^\prime = \left(q'_1 (1-T),..., q'_{k-1} (1-T), T, q'_{k+1} (1-T), q'_K (1-T)\right).$$

For each simulation scenario, we fix a threshold $\varepsilon_\infty$. To generate data inside the null hypothesis we choose individual admixtures $q_0$ with $T < \varepsilon_\infty$, where smaller values of $T$ correspond to a more extreme null hypothesis. Data is generated on the boundary by choosing $q_0$ with $T = \varepsilon_\infty$. Lastly, we simulate under the alternative by generating gene samples from an individual admixture $q_0$ with $T > \varepsilon_\infty$. Here, larger values of $T$ correspond to more extreme alternatives.

We consider scenarios with $K \in \{2,5\}$ and $M \in \{50, 100, 200, 500, 1000 \}$ and thresholds $0.63, 0.65$ and $0.75$. \new{For every scenario we evaluate the test on a grid of values of $T$ ranging from well inside the null hypothesis to $T = 0.99$.} We set the significance level to $\alpha = 0.05$, run the test with $B=100$ bootstrap repetitions and report the fraction of rejected null hypotheses from \new{$500$} simulation runs. This rejection rate is the empirical type $I$ error under the null hypothesis and the power under the alternative. We used $B = 100$ throughout, since empirical evaluation showed that increasing $B$ to $200$ or $1000$ does not influence the performance of the test much.

\begin{remark}[Discussion of the simulation scenarios]
    The case $K = 2$ is commonly considered in the Admixture Model \citep{heinzel2024revealing}. We choose $K = 5$ because this corresponds to the number of ancestral populations in the 1000 Genomes data \citep{10002015global} and is of the order of magnitude of reference panels used in forensic practice \citep{kidd2014,phillips2015}\new{. See also Remark~\ref{Rem:ChoiceOfK} for the role of $K$ in the supervised setting}. In forensic science, marker sets with $M \in [50, 200]$ are common \citep{kidd2014, ruiz2023development, maurer2025enhancing}, which motivates our choice of these values. We choose the thresholds $\varepsilon_\infty = 0.75$ and thresholds around  $\varepsilon_\infty = 5/8$ because these values correspond to the expected allele contribution under an idealized pedigree model with non-admixed ancestors of three out of four grandparents or five out of eight great-grandparents originating from the same population, respectively. \new{Higher thresholds, $\varepsilon_\infty \in \{0.75,0.85,0.95\}$, are studied in Section~\ref{sec:robustness}.}
\end{remark}

We start by discussing the simulation results for the ancestral allele frequencies following the Dirichlet$(1,1)$ distribution.
Figures \ref{fig:pp_K2_1} and \ref{fig:pp_K2_2} show that the type $I$ error inside the null hypothesis is close to zero across all considered values of $M$. Especially for values of $T$ close to the boundary $\varepsilon_\infty$ we see that the type $I$ error decreases when $M$ increases. This corresponds to the behavior proved in Theorem \ref{Th:BootstrapAsymptoticNormality} $(a)$.

At the boundary, that is when $T = \varepsilon_\infty$, the type $I$ error is generally close to the nominal level $\alpha = 0.05$. Table \ref{tab:first} shows the exact rejection rates on the boundary for all four scenarios, since they are important indicators for the performance of the test. \new{For $K=2$ the rejection rates lie between $0.044$ and $0.068$ for all considered values of $M$ and both thresholds, with a slightly liberal behaviour for the smallest marker set ($0.068$ for $M=50$ and $\varepsilon_\infty = 0.63$) and values close to $\alpha$ for large $M$. This is expected by the result in Theorem \ref{Th:BootstrapAsymptoticNormality} $(a)$.}

Lastly, Figures \ref{fig:pp_K2_1} and \ref{fig:pp_K2_2} show that for a fixed scenario $T$ the power increases as $M$ increases, as expected from Theorem \ref{Th:BootstrapAsymptoticNormality} $(b)$: for $K=2$, $\varepsilon_\infty = 0.63$ and $T=0.8$ the rejection rate grows from \new{around $0.92$ for $M= 200$ to $1$ for $M=500$ and $M=1000$}. For fixed $M$, the power increases the more extreme the alternative is, and it grows faster for large values of $M$ than for small ones. In the scenario $K=2$, $\varepsilon_\infty = 0.75$ \new{a power of $0.99$ is reached for $M=1000$ already at $T= 0.85$ and for $M=200$ at $T=0.95$}.

We observe the same patterns described above for $K = 5$ in Figures \ref{fig:pp_K5_1} and \ref{fig:pp_K5_2}.
\new{The rejection rates at the boundary for $K=5$, reported in Table \ref{tab:first}, lie between $0.042$ and $0.086$ and thus scatter around the nominal level without a systematic bias in either direction.}

\begin{figure}[h!]
    \centering
    \begin{subfigure}[b]{0.48\textwidth}
        \centering
        \includegraphics[width=\textwidth]{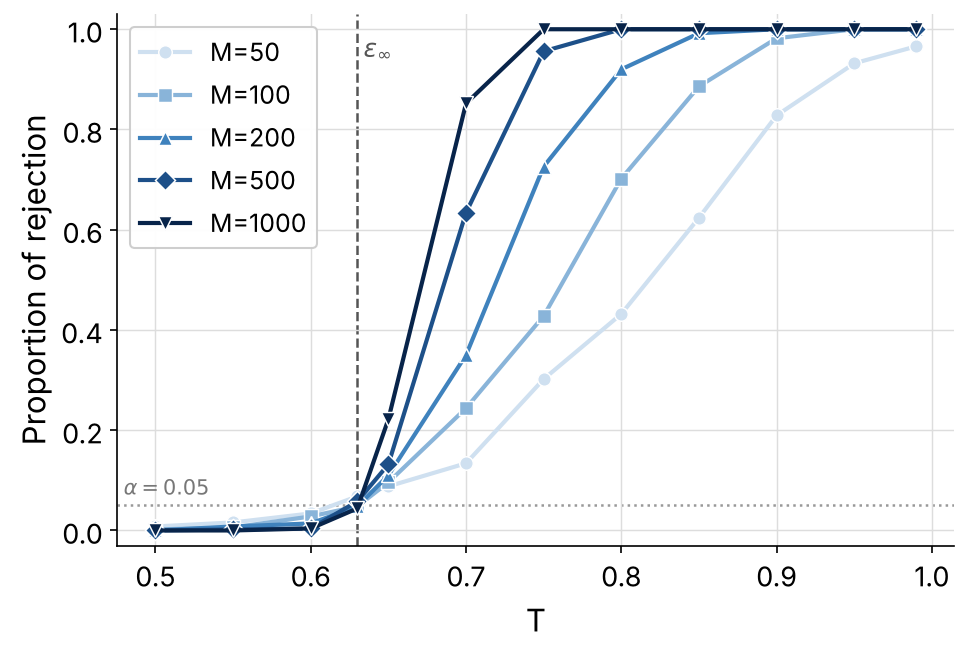}
        \caption{$K = 2$, $\varepsilon_\infty = 0.63$}
        \label{fig:pp_K2_1}
    \end{subfigure}
    \hfill
    \begin{subfigure}[b]{0.48\textwidth}
        \centering
        \includegraphics[width=\textwidth]{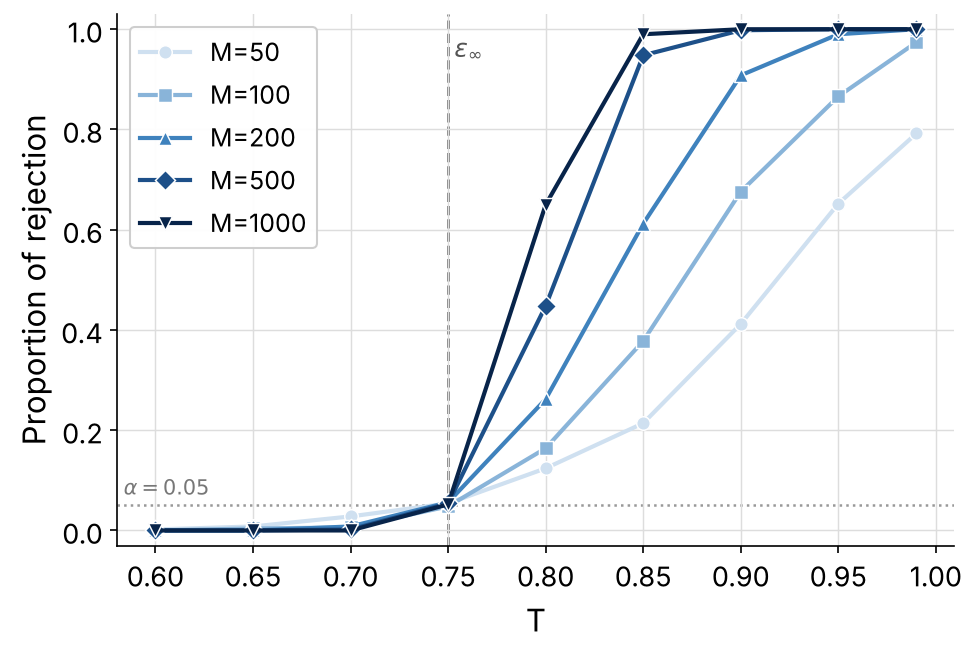}
        \caption{$K = 2$, $\varepsilon_\infty = 0.75$}
        \label{fig:pp_K2_2}
    \end{subfigure}

    \vspace{1em}

    \begin{subfigure}[b]{0.48\textwidth}
        \centering
        \includegraphics[width=\textwidth]{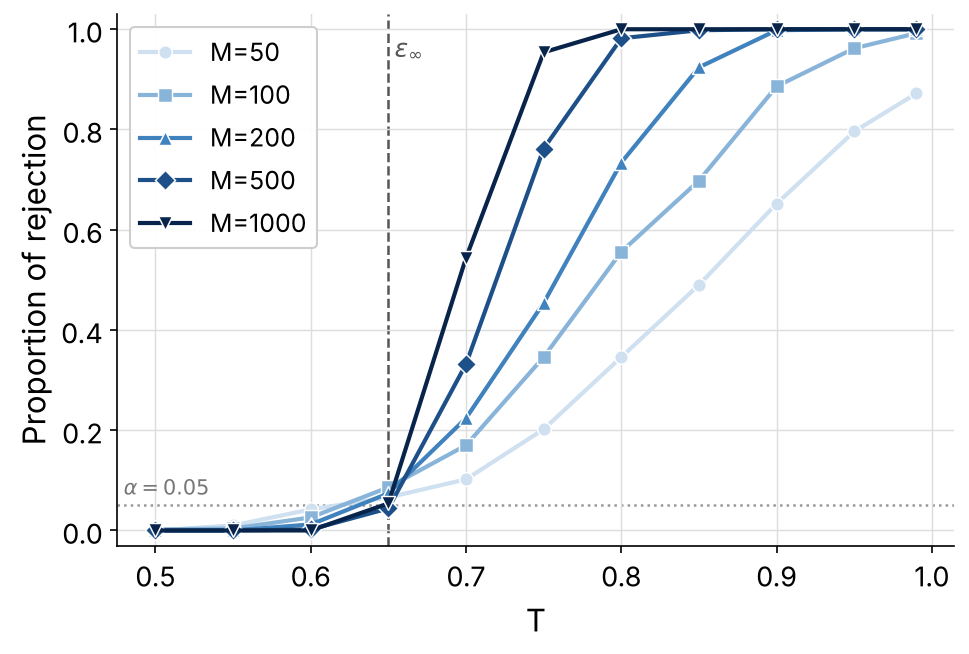}
        \caption{$K = 5$, $\varepsilon_\infty = 0.65$}
        \label{fig:pp_K5_1}
    \end{subfigure}
    \hfill
    \begin{subfigure}[b]{0.48\textwidth}
        \centering
        \includegraphics[width=\textwidth]{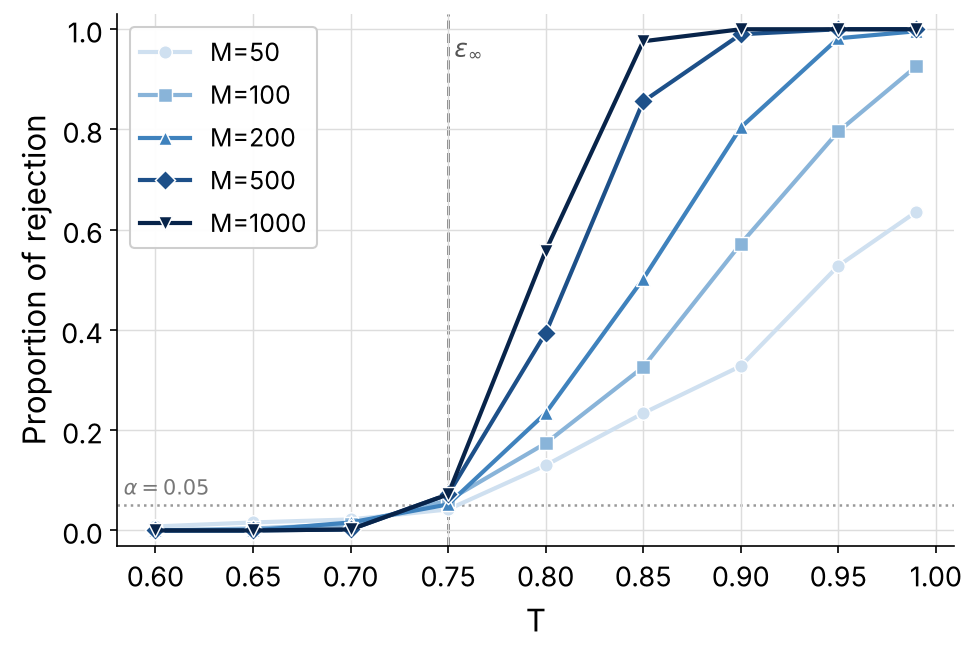}
        \caption{$K = 5$, $\varepsilon_\infty = 0.75$}
        \label{fig:pp_K5_2}
    \end{subfigure}
    \caption{Fraction of simulations in which the null hypothesis is rejected,
    for different values of $T$ and $M$. The first components of the allele
    frequencies are i.i.d. $\mathcal U([0,1])$-distributed.}
    \label{fig:generalEvaluation}
\end{figure}

\begin{table}[h!]
    \centering
\begin{tabular}{c|cc|cc}
\hline
 & \multicolumn{2}{c|}{$K = 2$} & \multicolumn{2}{c}{$K = 5$} \\
$M$ & $\varepsilon_\infty = 0.63$ & $\varepsilon_\infty = 0.75$
    & $\varepsilon_\infty = 0.65$ & $\varepsilon_\infty = 0.75$ \\
\hline
$50$   & 0.068 & 0.054 & 0.066 & 0.042 \\
$100$  & 0.048 & 0.048 & 0.086 & 0.060 \\
$200$  & 0.048 & 0.056 & 0.074 & 0.052 \\
$500$  & 0.058 & 0.054 & 0.044 & 0.072 \\
$1000$ & 0.044 & 0.052 & 0.054 & 0.072 \\
\hline
\end{tabular}
    \caption{Fraction of simulations in which the null hypothesis is rejected at
    the boundary $T = \varepsilon_\infty$, with i.i.d.
    $\mathcal U([0,1])$-distributed first components of the allele frequencies.
    Nominal level $\alpha = 0.05$, $B = 100$ bootstrap repetitions, \new{$500$}
    simulation runs per scenario.}
    \label{tab:first}\label{tab:second}
\end{table}

Replacing the $\mathcal U([0,1])$ design by a less informative $\beta(0.5,2)$ design lowers the power consistently, for example from $0.95$ to $0.81$ at $\varepsilon_\infty = 0.75$, $T = 0.85$ and $M=500$, while leaving the rejection rates at the boundary unaffected. Section~\ref{sec:robustness} studies this dependence systematically in terms of $F_{ST}$.

\begin{newpart}
\subsection{Robustness with respect to estimated allele frequencies and realistic population differentiation} \label{sec:robustness}

The theoretical analysis treats the ancestral allele frequencies $p$ as known, and the designs used above produce markers of very high ancestry information content. We therefore complement the study by an experiment closer to practice in both respects.

Allele frequencies are generated from the Balding--Nichols model \citep{baldingnichols1995}: for each marker $m$ an ancestral frequency $\pi_m \sim \mathcal U(0.05,0.95)$ is drawn, and the population-specific frequencies are drawn independently as
\begin{equation*}
    p_{k1m} \sim \beta\Big( \pi_m \tfrac{1-F_{ST}}{F_{ST}},
    (1-\pi_m) \tfrac{1-F_{ST}}{F_{ST}} \Big), \qquad k = 1,\dots,K,
\end{equation*}
for $F_{ST} \in \{0.025, 0.05, 0.1\}$. Where the reference frequencies are estimated, they are obtained from a panel of $n_{\textup{ref}}$ diploid individuals per population as $\hat p_{k1m} = (Y_{km} + 1/2)/(2 n_{\textup{ref}} + 1)$ with $Y_{km} \sim \textup{Bin}(2 n_{\textup{ref}}, p_{k1m})$, the pseudo-count preventing estimated frequencies of exactly $0$ or $1$ and thus preserving assumption $(A2)$. The individual is always simulated from the true frequencies $p$, so that the test is genuinely misspecified in that case. We choose $B=100$, $\alpha = 0.05$ and $100$ replicates per point. Figure~\ref{fig:fst} reports the rejection rates for known frequencies, Table~\ref{tab:robust_boundary} the comparison with $n_{\textup{ref}} \in \{50,500\}$ at $F_{ST}=0.025$, based on $200$ replicates per scenario.

Hence the test keeps its nominal level throughout: over all $108$ boundary scenarios the rejection rate has mean $0.059$ and never exceeds $0.12$, consistent with the Monte Carlo error of $0.022$ attached to $100$ replicates. Its validity is therefore not an artefact of the informative allele-frequency designs used above. The differentiation of the markers determines how many of them are needed, but not whether the procedure works: at $M=5000$ and $K=2$ the power at $T = \varepsilon_\infty + 0.10$ is $1.00$, $0.94$ and $0.78$ for $F_{ST} = 0.1$, $0.05$ and $0.025$, so halving $F_{ST}$ roughly doubles the number of markers required, and $K=5$ costs a further factor of about two. Table~\ref{tab:robust_boundary} additionally shows that estimating the allele frequencies makes the test conservative, at the price of power that is largely recovered for $n_{\textup{ref}} = 500$.
\begin{figure}[h!]
    \centering
    \includegraphics[width=\textwidth]{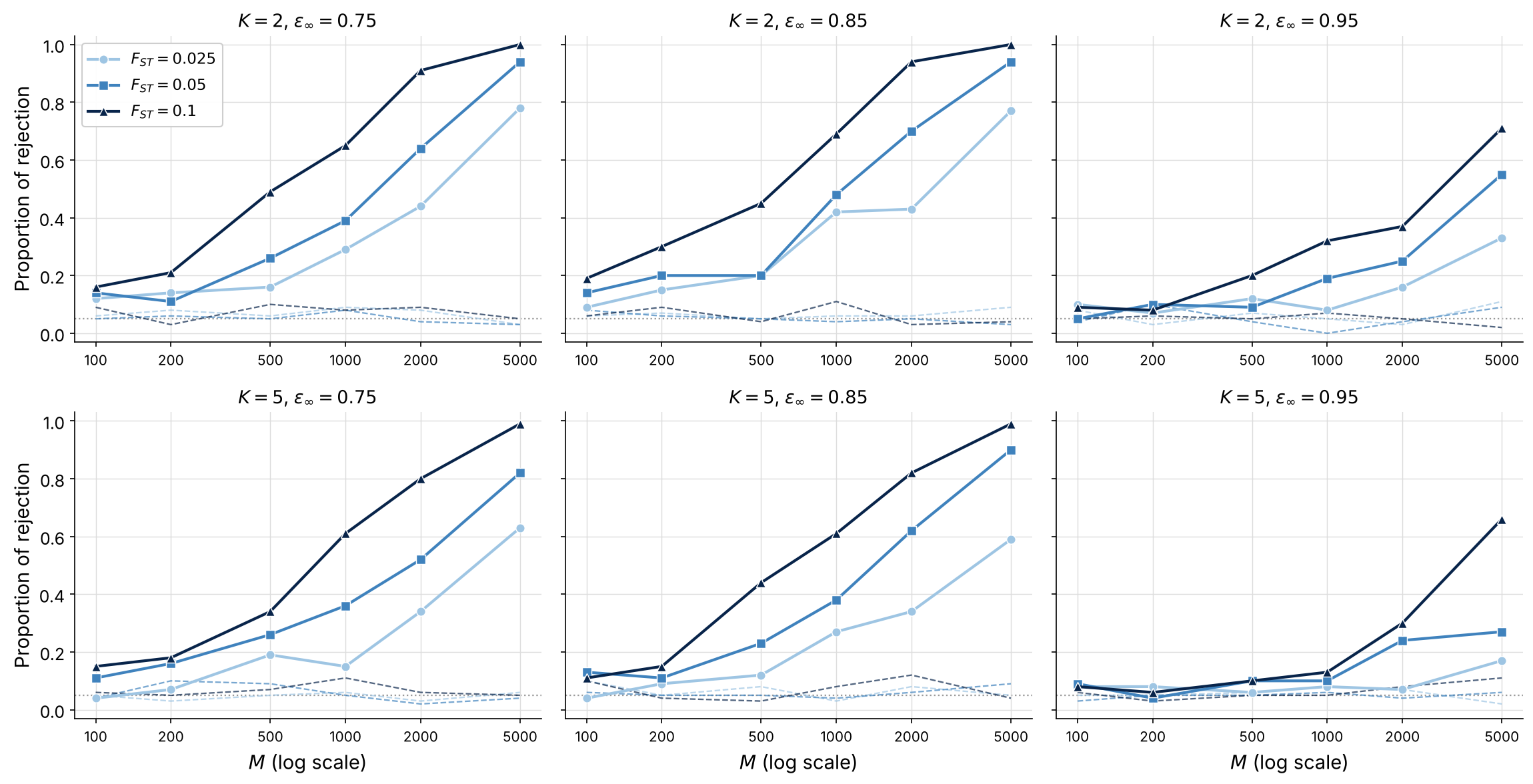}
    \caption{Rejection rates under the Balding--Nichols design with known allele frequencies, against the number of markers $M$ (logarithmic scale), for three levels of population differentiation. Solid lines show the power at the alternative $T = \varepsilon_\infty + 0.10$ (capped at $T=0.99$), dashed lines of the same colour the empirical level at the boundary $T = \varepsilon_\infty$, and the dotted horizontal line the nominal level $\alpha = 0.05$. Each point is based on $100$ simulation runs with $B = 100$ bootstrap repetitions.}
    \label{fig:fst}
\end{figure}

\begin{table}[h!]
    \centering
    \small
\begin{tabular}{ccc|ccc|ccc}
\hline
 & & & \multicolumn{3}{c|}{Level at $T = \varepsilon_\infty$}
     & \multicolumn{3}{c}{Power at $T = \varepsilon_\infty + 0.10$} \\
$K$ & $M$ & $\varepsilon_\infty$
 & \new{$\hat p$, $n_{\textup{ref}}{=}50$} & \new{$\hat p$, $n_{\textup{ref}}{=}500$} & \new{$p$ known}
 & \new{$\hat p$, $n_{\textup{ref}}{=}50$} & \new{$\hat p$, $n_{\textup{ref}}{=}500$} & \new{$p$ known} \\
\hline
$2$ & $100$ & $0.75$ & 0.035 & 0.030 & 0.075 & 0.070 & 0.050 & 0.150 \\
$2$ & $500$ & $0.75$ & 0.020 & 0.040 & 0.055 & 0.120 & 0.150 & 0.180 \\
$2$ & $100$ & $0.85$ & 0.030 & 0.040 & 0.025 & 0.040 & 0.085 & 0.090 \\
$2$ & $500$ & $0.85$ & 0.010 & 0.040 & 0.060 & 0.020 & 0.150 & 0.235 \\
$2$ & $100$ & $0.95$ & 0.015 & 0.055 & 0.055 & 0.010 & 0.065 & 0.070 \\
$2$ & $500$ & $0.95$ & 0.005 & 0.050 & 0.050 & 0.015 & 0.130 & 0.075 \\
\hline
$5$ & $100$ & $0.75$ & 0.020 & 0.050 & 0.045 & 0.045 & 0.070 & 0.090 \\
$5$ & $500$ & $0.75$ & 0.015 & 0.040 & 0.090 & 0.025 & 0.140 & 0.175 \\
$5$ & $100$ & $0.85$ & 0.025 & 0.075 & 0.055 & 0.020 & 0.070 & 0.060 \\
$5$ & $500$ & $0.85$ & 0.005 & 0.035 & 0.085 & 0.025 & 0.135 & 0.175 \\
$5$ & $100$ & $0.95$ & 0.015 & 0.030 & 0.050 & 0.015 & 0.055 & 0.080 \\
$5$ & $500$ & $0.95$ & 0.000 & 0.030 & 0.050 & 0.000 & 0.070 & 0.095 \\
\hline
\end{tabular}
    \caption{Empirical type~I error at the boundary $T = \varepsilon_\infty$ and
    empirical power at the alternative $T = \varepsilon_\infty + 0.10$ (capped at
    $T = 0.99$) under the Balding--Nichols design with $F_{ST} = 0.025$.
    \new{The genotypes are always generated from the true allele frequencies
    $p$, and the columns differ only in the reference frequencies the test is
    given: in the $\hat p$ columns these are estimated from a panel of
    $n_{\textup{ref}}$ diploid individuals per population, so that the test is
    misspecified, whereas in the ``$p$ known'' columns it receives the true $p$,
    which is the idealized setting of the theory.}
    Nominal level $\alpha = 0.05$, $B = 100$ bootstrap repetitions,
    $200$ simulation runs per scenario.}
    \label{tab:robust_boundary}\label{tab:robust_power}
\end{table}

\begin{remark}[Implementation and computational cost] \label{Rem:CostLargeM}\label{Rem:ComputationalCost}
We calculated the median time for one complete test decision with $B = 100$, measured on a single core of an Intel Xeon processor at $2.1$ GHz. For $K=2$ the cost of a single decision stays below one second up to $M = 10000$, and for $K=5$ it reaches about $15$ seconds there.

We infer from these results that a panel of $10^6$ markers would require on the order of minutes per individual for $K=2$ and roughly half an hour for $K=5$ on a single core, which is feasible for a single trace but demanding for a large cohort.

\end{remark}
\end{newpart}

\section{Application of the Statistical Test to Data}\label{sec:application}

\subsection{Application to the data from the 1000 Genomes Project}\label{sec:rev_1}

We applied the test to data from the 1000 Genomes Project \citep{10002015global}. For this purpose, we used the marker set of \cite{kidd2014}, which contains 55 bi-allelic markers, and included all \new{$2504$} individuals. We used the five $1000$ Genomes populations Europe (EUR), Africa (AFR), South Asia (SAS), East Asia (EAS), and Admixed Americans (AMR) as reference groups, so that $K = 5$.
We note that these reference groups should not be interpreted as truly non-admixed ancestral populations. In particular, AMR is itself an admixed reference group. The application is therefore intended as an illustration of the proposed testing procedure in a realistic reference-group setting\new{, and, in the terminology of Remark~\ref{Rem:ChoiceOfK}, all statements below are relative to this panel of five reference groups}. For each individual in the data set, we estimated the allele frequencies of the reference populations based on all other individuals in the data set\new{, so that no individual contributes to the frequencies against which it is tested. A pseudo-count of $1/2$ is added to the allele counts so that assumption $(A2)$ is not violated by sampling zeros}. Thus, note that this application goes beyond the idealized theoretical setting of known allele frequencies. \new{Section~\ref{sec:robustness} suggests that the effect of this is conservative, and that with reference groups of several hundred individuals it is small}. Based on these reference allele frequencies we applied the testing algorithm presented in Algorithm \ref{alg} for each individual with $B=100$ bootstrap repetitions and predominance threshold $\varepsilon_\infty = 0.75$. The number of individuals per population for which the null hypothesis in \eqref{test:swapped} is rejected \new{is shown by the light bars in the left panel of Figure \ref{fig:application}. The dark bars belong to the sensitivity analysis of Section~\ref{sec:K4} and may be ignored for the moment}.

\begin{figure}[h!]
    \centering
    \includegraphics[width=\textwidth]{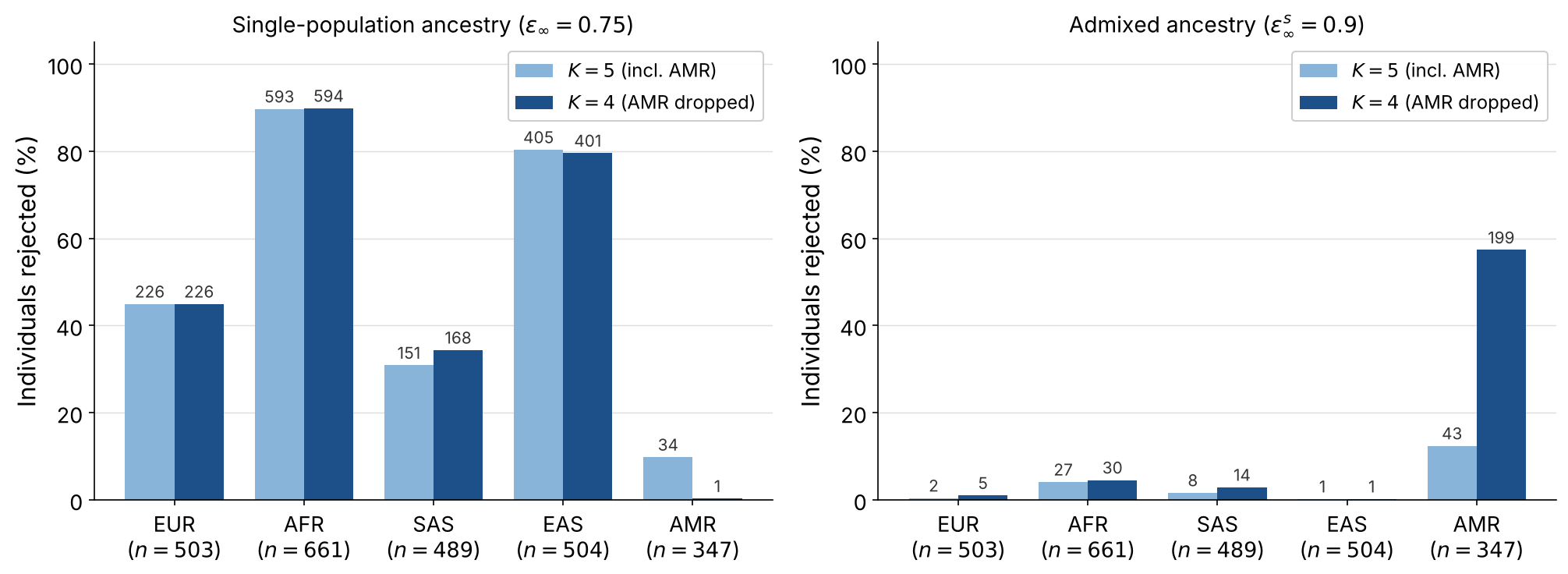}
    \caption{\new{Results for the 1000 Genomes data with the markers from \cite{kidd2014}: percentage of individuals
    per population for which the null hypothesis is rejected, with the
    corresponding number of individuals above each bar and the group size
    $n$ on the horizontal axis. The light bars belong to the
    reference panel with $K=5$ populations (this section) and the dark bars to
    the panel with $K=4$, obtained by dropping AMR (Section~\ref{sec:K4}). The
    left panel shows test \eqref{test:swapped} with
    $\varepsilon_\infty = 0.75$ and the right panel test
    \eqref{HypothesisAdmixtureOneComponent} with
    $\varepsilon_\infty^s = 0.9$. The horizontal axis always refers to the
    original population label, so under $K=4$ the AMR individuals are tested
    against a panel that does not contain their own group.}}
    \label{fig:application}
\end{figure}

Figure \ref{fig:application} shows substantial differences between populations. \new{For individuals from AFR and EAS, the null hypothesis is usually rejected, with rejection rates of around $90\%$ and $80\%$, respectively. For individuals from AMR, the rejection rate is small with around $10\%$. We also observe that for individuals from EUR and SAS the null hypothesis is rejected for roughly $45\%$ and $31\%$ of the individuals, respectively. Over the whole data set, the null hypothesis is rejected for $1409$ of the $2504$ individuals, that is for about $56\%$.} Intuitively, rejecting the null hypothesis of \eqref{test:swapped} means that the individual has a high percentage of its genome from one single reference population. Hence, it is not surprising that for example population AMR has a low fraction of individuals for which the null hypothesis is rejected. The AMR population is itself admixed and heterogeneous, so AMR individuals are less likely to be represented by one predominant reference component. In contrast, groups such as AFR and EAS are more strongly differentiated by the marker set and therefore more often yield estimated admixture vectors with one clearly predominant component.

Analogously to how we apply test \eqref{test:swapped} to the data from \cite{10002015global} above, we  also apply test \eqref{HypothesisAdmixtureOneComponent}. The results with $B=100$ bootstrap repetitions and threshold $\varepsilon_\infty^s = 0.9$ are shown \new{in the right panel of Figure \ref{fig:application}}. The threshold $\varepsilon_\infty^s = 0.9$ corresponds to evidence against a very strong single-population contribution. \new{The null hypothesis is rejected for $81$ individuals, approximately $3\%$ of the data set, with substantial differences between populations: approximately $12\%$ for AMR against about $0.2\%$ for EAS.} This is again expected because AMR is itself an admixed and heterogeneous reference group, whereas the low rejection rate for EAS indicates that most EAS individuals have one highly predominant reference component with respect to the chosen marker set and reference populations. Both hypothesis tests therefore show similar patterns for the population groups.

\begin{newpart}
\subsection{Sensitivity with respect to the choice of the reference panel}
\label{sec:K4}

As discussed in Remark~\ref{Rem:ChoiceOfK}, the hypotheses in
\eqref{test:swapped} are statements relative to the reference panel, so it is
natural to ask how sensitive the conclusions above are to its composition. One
of the five 1000 Genomes populations used there, AMR, is itself admixed and
therefore only imperfectly represents a non-admixed ancestral population. We
repeat the analysis with a reference panel of $K = 4$ groups obtained by
dropping AMR, keeping EUR, AFR, SAS and EAS. All $2504$ individuals are still
tested. The outcome is given by the dark bars
in Figure~\ref{fig:application}.

For individuals belonging to the four retained reference groups, the aggregate rejection proportions for the test in \eqref{test:swapped} change relatively little, from $44.9\%$ to $44.9\%$ for EUR, from $89.7\%$ to $89.9\%$ for AFR, from $80.4\%$ to $79.6\%$ for EAS and from $30.9\%$ to $34.4\%$ for SAS. The modest increase for SAS may reflect a redistribution of fitted ancestry mass towards a leading reference component for some individuals.

For the AMR individuals the picture changes completely.
Under $K=5$ the null hypothesis of \eqref{test:swapped} was rejected for
$9.8\%$ of them. Under $K=4$ this drops to $0.3\%$. Conversely, the test for admixed ancestry in
\eqref{HypothesisAdmixtureOneComponent} rejects its null hypothesis for $57.3\%$
of the AMR individuals under $K=4$, compared with $12.4\%$ under $K=5$. The predominant component is most often EUR ($203$ individuals). This pattern is compatible with the heterogeneous ancestry of AMR individuals, which commonly includes European, Indigenous American and African contributions.
\end{newpart}

\section{Discussion}
\label{sectionDiscussion}

\new{In this article, we develop a constrained parametric bootstrap test for whether the largest component of an individual's admixture vector exceeds a prespecified predominance threshold. Rejection provides positive evidence for predominant ancestry relative to the specified reference populations while allowing secondary ancestry contributions. Complementing existing results on individual-admixture estimation \citep{pfaffelhuber2022central,heinzel2025consistency}, we establish pointwise asymptotic type-I error control and consistency under the stated model assumptions.}

\new{The supervised setting is particularly appropriate for applications in which an individual's ancestry is assessed relative to an externally established reference panel, since this provides a fixed population representation for interpreting the ancestry proportions, without the investigated individual's genotype influencing the reference frequencies used for their assessment. Within this framework, our theoretical analysis treats the reference allele frequencies as known, and the simulation study in Section~\ref{sec:robustness} investigates the effect of estimating them from finite reference samples: smaller samples generally lead to lower rejection rates at the null boundary and reduced power, whereas larger ones yield results closer to the known-frequency benchmark.} A formal asymptotic theory that accounts for the additional randomness of estimated reference frequencies, \new{possibly under a joint regime in which both the number of markers and the reference panel sizes grow,} remains an important direction for future work. \new{A joint bootstrap that accounts for reference-sample variability alongside the individual genotype data could provide a possible approach.}

Although a Wald-type test could be constructed from the asymptotic normality of the MLE, it would require estimating the variance of the maximum ancestry component and would depend on a normal approximation. Because the number of markers may be limited and the null hypothesis involves a threshold condition with a least-favorable boundary case, we instead employ a constrained parametric bootstrap, \new{which approximates the null distribution of the test statistic by generating data under a null-constrained estimator. The procedure also avoids testing each ancestry component separately: a component-wise approach would require a suitable calibration of the combined testing decisions, potentially involving multiple-testing corrections such as the Bonferroni adjustment \citep{sainani2009problem,vanderweele2019some}, which can reduce power, whereas our procedure directly calibrates the global hypothesis concerning the maximum ancestry component.}

From a mathematical standpoint, our results are noteworthy because they extend the theory of \cite{DetteMoellenhoff} to \new{independent, non-identically distributed observations}.
We expect that both our proof techniques and our implementation can be adapted to related hypothesis-testing problems involving independent but non-identically distributed random variables.

\new{The reference-panel sensitivity analysis in Section~\ref{sec:K4} illustrates how the choice of reference populations influences ancestry estimates and testing decisions: removing AMR has relatively little effect on the aggregate rejection proportions for the four retained groups, whereas AMR individuals exhibit less concentrated estimated ancestry proportions on average and substantially fewer predominance declarations. This emphasizes the importance of interpreting the results relative to the specified reference populations, as discussed in Remark~\ref{Rem:ChoiceOfK}. Whether the panel contains an appropriate population for a given trace is a complementary question, addressed by procedures such as those of \citet{tvedebrink2018, tvedebrink2019inference}, and we regard the combination of such a panel-adequacy check with the present test as the natural workflow in practice.}

It would also be valuable to investigate whether the theoretical framework extends to linkage models \citep{falush2003}. Such an extension would require substantially different techniques, as the independence assumption across markers would no longer hold.

\noindent
\textbf{Funding} \\
\noindent
This project received funding by the Deutsche Forschungsgemeinschaft (DFG, German Research Foundation) – Project-ID 499552394 – SFB Small Data. HD is supported by the European Union through the European Joint Programme on Rare Diseases under the European Union’s Horizon 2020 Research and Innovation Programme Grant Agreement Number 825575. He is task leader of RealiseD supported by the Innovative Health Initiative Joint Undertaking (IHI JU) under grant agreement No 101165 912. The JU
receives support from the European Union’s Horizon Europe research and innovation programme and COCIR, EFPIA, Europa Bío, MedTech Europe, and Vaccines Europe.\\
\noindent
\new{
\textbf{Acknowledgment} \\
\noindent
The authors thank two anonymous reviewers for the valuable comments that lead, among other things, to a better choice of the simulation scenarios and a better description what our method provides compared to previously existing approaches. \\}

\noindent
\textbf{Declaration of Conflicts of Interests} \\
\noindent
The authors declare that there are no conflicts of interest. \\
\new{
\textbf{Declaration of the Use of AI} \\
\noindent
The authors declare that they used Claude to write parts of the python code, to complement their literature review and to improve the formulations in this paper. }
\bibliography{literatur,CBB}

\end{document}